\documentclass{LMCS}

\def\doi{9(2:04)2013}
\lmcsheading%
{\doi}
{1--37}
{}
{}
{Aug.~17, 2012}
{May.~22, 2013}
{}

\usepackage{enumerate}
\usepackage{hyperref}
\usepackage{amssymb,amsmath,amsfonts}
\usepackage{xspace}
\usepackage{ltlfonts}
\usepackage[section]{algorithm}
\usepackage{algorithmic}
\usepackage{tikz}
\usetikzlibrary{arrows}
\usetikzlibrary{automata}
\usetikzlibrary{shapes}

\usepackage{amsmath,amssymb,amsfonts,mathrsfs,latexsym,stmaryrd}
\usepackage{array}
\usepackage{url}
\usepackage{comment}
\usepackage{epsfig}
\usepackage{ifpdf}

\ifpdf
\else
\fi

\makeatletter
\def\shortrightarrowfill@{\arrowfill@\relbar\relbar\shortrightarrow}
\newcommand{\ort}{\mathpalette{\overarrow@\shortrightarrowfill@}}
\makeatother
\makeatletter
\def\shortleftarrowfill@{\arrowfill@\relbar\relbar\shortleftarrow}
\newcommand{\olft}{\mathpalette{\overarrow@\shortleftarrowfill@}}
\makeatother
\makeatletter
\def\shortleftrightarrowfill@{\arrowfill@\relbar\relbar\leftrightarrow}
\newcommand{\olftrt}{\mathpalette{\overarrow@\shortleftrightarrowfill@}}
\makeatother

\def\CC{\bf}
\newcommand{\myeat}[1]{}
\newcommand{\kw}[1]{{\ensuremath {{\mathsf{#1}}}}\xspace}
\newcommand{\subf}{\kw{sub}}

\newcommand{\olet}{\kw{Let}}
\newcommand{\obe}{\kw{be}}
\newcommand{\oin}{\kw{in}}

\def\CIRC{\LTLcircle}
\def\CIRCM{\LTLcircleminus}
\def\DIAMOND{\LTLdiamond}
\def\DIAMONDM{\LTLdiamondminus}
\def\BOX{\LTLsquare}

\def\TL{\mathrm{TL}[\DIAMOND,\DIAMONDM]}
\def\TLlet{\mathrm{TL[\DIAMOND,\DIAMONDM]_{\kw Let}}}
\def\UTL{\mathrm{UTL}}
\def\UTLlet{\mathrm{UTL_{\kw Let}}}
\def\LTL{\mathrm{LTL}}
\def\LTLlet{\mathrm{LTL_{\kw Let}}}

\def\FOtwo{\mathrm{FO^2}}
\def\FOtwoLT{\mathrm{FO^2[<]}}
\def\FOtwoLTL{\mathrm{FO^2[LTL]}}
\def\FOtwoLET{\mathrm{FO^2_{\kw Let}}}
\def\FOtwoLTlet{\mathrm{FO^2[<]_{\kw Let}}}
\def\FOtwoLTLlet{\mathrm{FO^2[LTL]_{\kw Let}}}

\def\FO2{{\rm FO}^2}
\def\UTL{{\mathrm{UTL}}}
\def\LTL{{\mathrm{LTL}}}

\def\vs{\boldsymbol{s}}
\def\vt{\boldsymbol{t}}
\def\vu{\boldsymbol{u}}

\def\~{\sim}
\def\->{\rightarrow}

\newif\ifpdf
\ifx\pdfoutput\undefined
  \pdffalse
\else
  \pdfoutput=1 \pdftrue
\fi

\catcode`\@=12

\newcommand{\fotwo}{\mathrm{FO^2}}

\newcommand{\myparagraph}[1]{{\bf #1.}}
\makeatletter

\pdfmapfile{+ltlfonts.map}

\def\cal{\mathcal}
\def\M{\mathcal{M}}
\def\order{\alpha}

\begin{document}

\title[Two Variable vs. Linear Temporal Logic in Model Checking and Games]{Two Variable vs. Linear Temporal Logic in Model Checking and Games\rsuper*}

\author[M.~Benedikt]{Michael Benedikt}	

\address{Department of Computer Science, University of Oxford, United Kingdom}	
\email{\{michael.benedikt, rastislav.lenhardt, jbw\}@cs.ox.ac.uk}

\author[R.~Lenhardt]{Rastislav Lenhardt}	
\address{\vskip-6 pt}	

\author[J.~Worrell]{James Worrell}	
\address{\vskip-6 pt}	


\ACMCCS{[{\bf Theory of computation}]: Logic---Verification by model
  checking; Formal languages and automata theory; Models of
  computation---Abstract machines}


\titlecomment{{\lsuper*}This includes material presented in Concur
  2011 and QEST 2012 extended abstracts}

\begin{abstract}
Model checking linear-time properties expressed in first-order logic
has non-elementary complexity, and thus various restricted logical
languages are employed.  In this paper we consider two such restricted
specification logics, linear temporal logic (LTL) and two-variable
first-order logic ($\FOtwo$).  LTL is more expressive but $\FOtwo$ can
be more succinct, and hence it is not clear which should be easier to
verify.  We take a comprehensive look at the issue, giving a
comparison of verification problems for $\FOtwo$, LTL, and various
sub-logics thereof across a wide range of models.  In particular, we
look at unary temporal logic (UTL), a subset of LTL that is
expressively equivalent to $\FOtwo$; we also consider the stutter-free
fragment of $\FOtwo$, obtained by omitting the successor relation, and
the expressively equivalent fragment of UTL, obtained by omitting the
next and previous connectives.

We give three logic-to-automata translations which can be used to give
upper bounds for $\FOtwo$ and UTL and various sub-logics.  We apply these
to get new bounds for both non-deterministic systems (hierarchical and
recursive state machines, games) and for probabilistic systems (Markov
chains, recursive Markov chains, and Markov decision processes). We
couple these with matching lower-bound arguments.

Next, we look at combining $\FOtwo$ verification techniques with those
for LTL. We present here a language that subsumes both $\FOtwo$ and
LTL, and inherits the model checking properties of both languages.
Our results give both a unified approach to understanding the
behaviour of $\FOtwo$ and LTL, along with a nearly comprehensive
picture of the complexity of verification for these logics and their
sublogics.
\end{abstract}

\maketitle

\section{Introduction}
The complexity of verification problems clearly depends on the
specification language for describing properties.  Arguably the most
important such language is \emph{Linear Temporal Logic} (LTL).  LTL
has a simple syntax, one can verify LTL properties over Kripke
structures in polynomial space, and one can check satisfiability also
in polynomial space.  Moreover, Kamp~\cite{phd-kamp} has shown that LTL
has the same expressiveness as first-order logic over words.  For
example, the first-order property ``after we are born, we live until
we die'':
$$ \forall x\mathrm{~} (born(x) \rightarrow \exists y \ge x\mathrm{~}
die(y) \wedge \forall z\mathrm{~} (x\le z<y \rightarrow live(z) ))$$
is expressed in LTL by the formula $\BOX (born \rightarrow
live \mathrel{\mathcal{U}} die)$.

In contrast with LTL, model checking first-order queries has
non-elementary complexity~\cite{phd_stockmeyer}---thus LTL could be
thought of as a tractable syntactic fragment of FO.  Another approach
to obtaining tractability within first-order logic is by maintaining
first-order syntax, but restricting to two-variable formulas.  The
resulting specification language $\FO2$ has also been shown to have
dramatically lower complexity than full first-order logic.  In
particular, Etessami, Vardi and Wilke~\cite{fo2_utl} showed that
satisfiability for $\FO2$ is NEXPTIME-complete and that $\FO2$ is
strictly less expressive than FO (and thus less expressive than LTL
also).  Indeed, \cite{fo2_utl} shows that $\FO2$ has the same
expressive power as \emph{Unary Temporal Logic} (UTL): the fragment of
LTL with only the unary operators ``previous'', ``next'', ``sometime
in the past'', ``sometime in the future''. Consider the example above.
We have shown that it
can be expressed in  LTL, but it is easy to show that
it cannot be expressed in UTL, and therefore cannot be expressed
in $\FO2$.

Although $\FO2$ is less expressive than LTL, there are some properties
that are significantly easier to express in $\FO2$ than in LTL.
Consider the property that two $n$-bit identifiers agree:
$$\exists x \, \exists y \, (x < y \wedge \bigwedge_{1\le i \le n}
b_i(x) \leftrightarrow b_i(y)) \, .$$ 
It is easy to show that there is an exponential blow-up in
transforming the above $\FO2$ formula into an equivalent LTL
formula. We thus have three languages $\UTL$, $\LTL$ and $\FO2$, with
$\UTL$ and $\FO2$ equally expressive, $\LTL$ more expressive, and with
$\FO2$ incomparable in succinctness with LTL.

Are verification tasks easier to perform in LTL, or in $\FO2$? This is
the main question we address in this paper. There are well-known examples
of problems that are easier in LTL than in $\FO2$: in particular
satisfiability, which is PSPACE-complete for LTL and NEXPTIME-complete
for $\FO2$~\cite{fo2_utl}.  We will show that there are also tasks
where $\FO2$ is more tractable than LTL.

Our main contribution is a uniform approach to the verification of
$\FO2$ via automata.  We show that translations to the appropriate
automata can give optimal bounds for verification of $\FO2$ on both
non-deterministic and probabilistic structures.  We also show that
such translations allow us to understand the verification of the
fragment of $\fotwo$ formed by removing the successor relation from
the signature, denoted $\FOtwoLT$.  It turns out, somewhat
surprisingly, that for this fragment we can get the same complexity
upper bounds for verification as for the simplest temporal
logic---$\TL$.  For our translations from $\FOtwoLT$ to automata, we
make use of a key result from Weis~\cite{WeisPhd}, showing that models
of $\FOtwoLT$ formulas realise only a polynomial number of types.  We
extend this ``few types'' result from finite to infinite words and use
it to characterise the structure of automata for $\FOtwoLT$.

The outcome of our translations is a comprehensive analysis of the
complexity of $\FOtwo$ and UTL verification problems, together with
those for the respective stutter-free fragments $\FOtwoLT$ and
$\TL$. We begin with model checking problems for Kripke structures and
for recursive state machines (RSMs), which we compare to known results
for LTL on these models. We then turn to two-player games, considering
the complexity of the problem of determining which player has a
strategy to ensure that a given formula is satisfied.  We then move
from non-deterministic systems to probabilistic systems.  We start
with Markov chains and recursive Markov chains, the analogs of Kripke
structures and RSMs in the probabilistic case. Finally we consider
one-player stochastic games, looking at the question of whether the
player can devise a strategy that is winning with a given probability.

Towards the end of the paper, we consider extensions of $\FOtwo$, and
in particular how $\FOtwo$ verification techniques can be combined
with those for Linear Temporal Logic (LTL).  We present here a
language that we denote $\FOtwoLTL$, subsuming both $\FOtwo$ and
LTL. We show that the complexity of verification problems for
$\FOtwoLTL$ can be attacked by our automata-theoretic methods, and
indeed reduces to verification of $\FOtwo$ and LTL individually. As a
result we show that the worst-case complexity of probabilistic
verification, as well as non-deterministic verification, for
$\FOtwoLTL$ is (roughly speaking) the maximum of the complexity for
$\fotwo$ and LTL.

This paper expands on results presented in two conference
papers, \cite{BLW, BLWqest}.

{\bf Organization:} Section \ref{sec:prelims} contains preliminaries,
while Section \ref{sec:infrastructure} gives fundamental results on
the model theory of $\FO2$ and its relation to $\UTL$ that will be
used in the remainder of the paper.  
%
Section \ref{sec:trans} presents the logic-to-automata
translations used in our upper bounds. The first is a translation of a
given $\UTL$ formula to a large disjoint union of B\"{u}chi automata
with certain structural restrictions.  This can also be used to give a
translation from a given $\FO2$ formula to an (still larger) union of
B\"{u}chi automata.  The second does something similar for $\FOtwoLT$
formulas.  The last translation maps $\FOtwoLT$ and $\FO2$ formulas to
deterministic parity automata, which is useful for certain problems
involving games.

Section~\ref{sec:nondet} gives upper and lower bounds for
non-deterministic systems, while Section~\ref{sec:prob} is concerned with
probabilistic systems. In Section~\ref{sec:fo2ltl} we consider model
checking of $\FOtwoLTL$, which subsumes both $\FOtwo$ and LTL, and
finally in Section~\ref{sec:let} we consider the impact of extending
all the previous logics with \emph{let definitions}.

\section{Logic, Automata and Complexity Classes} \label{sec:def}
\label{sec:prelims}
We consider a first-order signature with set of unary predicates
$\mathcal{P}=\{P_1,\ldots,P_m\}$ and binary predicates $<$ (less than) and
$\mathrm{suc}$ (successor). Fixing two distinct variables $x$ and $y$,
we denote by $\FOtwo$ the set of first-order formulas over the above
signature involving only the variables $x$ and $y$.  We denote by
$\FOtwoLT$ the sublogic in which the binary predicate $\mathrm{suc}$
is not used.  We write $\varphi(x)$ for a formula in which only the
variable $x$ occurs free.

In this paper we are interested in interpretations of $\FOtwo$ on
infinite words.  An $\omega$-word $u=u_0u_1\ldots$ over the powerset
alphabet $\Sigma = 2^{\mathcal{P}}$ represents a first-order structure
extending $\langle \mathbb{N},<,\mathrm{suc} \rangle$, in which
predicate $P_i$ is interpreted by the set $\{ n \in \mathbb{N} : P_i
\in u_n\}$ and the binary predicates $<$ and $\mathrm{suc}$ have the
obvious interpretations.

We also consider Linear Temporal Logic
$\LTL$ on $\omega$-words.  The formulas of $\LTL$ are built from
atomic propositions using Boolean connectives and the temporal
operators $\CIRC$ (\emph{next}), $\CIRCM$ (\emph{previously}),
$\DIAMOND$ (\emph{eventually}), $\DIAMONDM$~(\emph{sometime in the
past}), $\mathcal{U}$ (\emph{until}), and $\mathcal{S}$
(\emph{since}).  Formally, {$\LTL$} is defined by the following
grammar:
\begin{eqnarray*}
 \varphi &::=& P_i \, \mid \, \varphi \wedge \varphi \,
                    \mid \, \neg \varphi \, 
		     \mid \, \varphi \mathrel{\mathcal{U}} \varphi \,
   		    \mid \, \varphi \mathrel{\mathcal{S}} \varphi \,
                    \mid \, \mathop{\DIAMOND} \varphi \, \mid 
                     \, \mathop{\DIAMONDM} \varphi \,
                    \mid \, \mathop{\CIRC} \varphi \,
                    \mid \, \mathop{\CIRCM} \varphi \, ,
\end{eqnarray*}
where $P_0, P_1, \ldots$ are propositional variables.  Unary temporal
logic ($\UTL$) denotes the subset without $\mathcal{U}$ and
$\mathcal{S}$, while $\TL$ denotes the
\emph{stutter-free} subset of $\UTL$ without $\CIRC$ and $\CIRCM$.
We use $\BOX \varphi$ as an abbreviation for $\neg \DIAMOND \neg \varphi$.

Let $(u, i)$ be the suffix $u_iu_{i+1}\ldots$ of $\omega$-word $u$. We define the semantics of $\LTL$ inductively
on the structure of the formulas as follows:
\begin{enumerate}[(1)]
\item $(u, i) \models P_k$ iff atomic prop. $P_k$ holds at position $i$ of $u$
\item $(u, i) \models \varphi_1 \wedge \varphi_2$ iff $(u, i) \models \varphi_1$ and $(u, i) \models \varphi_2$
\item $(u, i) \models \neg \varphi$ iff it is not the case that $(u, i) \models \varphi$
\item $(u, i) \models \mathop{\CIRC} \varphi$ iff $(u, i+1) \models \varphi$
\item $(u, i) \models \mathop{\CIRCM} \varphi$ iff $(u, i-1) \models \varphi$
\item $(u, i) \models \varphi_1 \mathrel{\mathcal{U}} \varphi_2$ iff $\exists j\ge i$ s.t. $(u, j) \models \varphi_2$ and $\forall k$, $i\le k < j$ we have $(u, k) \models \varphi_1$
\item $(u, i) \models \varphi_1 \mathrel{\mathcal{S}} \varphi_2$ iff $\exists j\le i$ s.t. $(u, j) \models \varphi_2$ and $\forall k$, $j < k \le i$ we have $(u, k) \models \varphi_1$
\item $(u, i) \models \mathop{\DIAMOND} \varphi$ iff $(u, i) \models \text{true} \mathrel{\mathcal{U}} \varphi$
\item $(u, i) \models \mathop{\DIAMONDM} \varphi$ iff $(u, i) \models \text{true} \mathrel{\mathcal{S}} \varphi$
\end{enumerate}

It is well known that over $\omega$-words $\LTL$ has the same
expressiveness as first-order logic, and $\UTL$ has the same
expressiveness as $\FOtwo$.  Moreover, while $\FOtwo$ is less
expressive than $\LTL$, it can be exponentially more
succinct~\cite{fo2_utl} -- for concrete examples of these facts,
 see the introduction.  

\begin{figure}[h!]
\begin{center}
\scalebox{0.9}{

\includegraphics{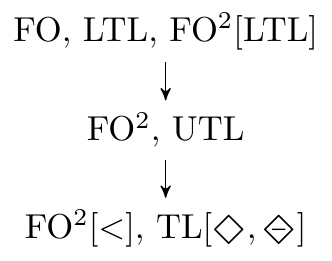}

}
\end{center}
\caption{Expressiveness Diagram}
\label{fig:express}
\end{figure}

We can combine the succinctness of $\FOtwo$ and the expressiveness of
$\LTL$ by extending the former with the temporal operators
$\mathcal{U}$ and $\mathcal{S}$ (applied to formulas with at most one
free variable).  We call the resulting logic $\FOtwoLTL$.  The syntax
of $\FOtwoLTL$ divides formulas into two syntactic classes:
\emph{temporal formulas} and \emph{first-order formulas}.  Temporal
formulas are given by the grammar
\begin{align*}
\varphi & ::=  P_i \, \mid \, \varphi \wedge \varphi \, \mid \, \neg \varphi \, \mid \, \varphi \mathrel{\mathcal{U}} \varphi \, \mid \,
\varphi \mathrel{\mathcal{S}} \varphi \, \mid \, \psi \, ,
\end{align*}
where $P_i$ is an atomic proposition and $\psi$ is a first-order
formula with one free variable.  First-order formulas are given by the
grammar
\begin{align*}
\psi & ::= \varphi(x) \, \mid \, x<y \, \mid \, \mathrm{suc}(x,y) \, \mid \,
\psi \wedge \psi \, \mid \, \neg \psi \, \mid \, \exists x \, \psi \, ,
\end{align*}
where $\varphi$ is a temporal formula.  Here the first-order formula
$\varphi(x)$ asserts that the temporal formula $\varphi$ holds at
position $x$.  The temporal operators $\CIRC$, $\CIRCM$, $\DIAMOND$
and $\DIAMONDM$ can all be introduced as derived operators.
An example of $\FOtwoLTL$ formula is:
\[b_0 \mathrel{\mathcal{U}} (\exists y \, (y < x \wedge \bigwedge_{1\le i \le n} b_i(x) \leftrightarrow b_i(y))) \, .\]
The relative expressiveness of the logics defined thus far is
illustrated in Figure~\ref{fig:express}.

Finally, we consider an extension of $\FOtwoLTL$ with \emph{let
  definitions}.  We inductively define the formulas and the unary
predicate subformulas that occur \emph{free} in such a formula.  The
atomic formulas of $\FOtwoLTLlet$ are as in $\FOtwoLTL$, with the
formula $P(x)$ occurring freely in itself.  The constructors include
all those of $\FOtwoLTL$, with the set of free subformula occurrences
being preserved by all of these constructors.

There is one new formula constructor  of the form:
\begin{eqnarray*}
\varphi & ::= & \olet ~ P_i(x) ~ \obe ~ \varphi_1(x) ~ \oin ~ \varphi_2
\end{eqnarray*}
where $P_i$ is a unary predicate, $\varphi_1(x)$ is an $\FOtwoLTLlet$
formula in which $x$ is the only free variable and no occurrence of
predicate $P_i$ is free, and $\varphi_2$ is an arbitrary
$\FOtwoLTLlet$ formula.  A subformula $P_j(z)$ occurs freely in
$\varphi(x)$ iff it occurs freely in $\varphi_1(x)$ or it occurs
freely in $\varphi_2$ and the predicate is not $P_i$.

The semantics of $\FOtwoLTLlet$ is defined via a translation function
$T$ to $\FOtwoLTL$, with the only non-trivial rule being:
$$T(\olet ~  P_i(x) ~ \obe ~ \varphi_1(x) ~ \oin ~ \varphi_2)  ::= $$
$$T(\varphi_2[P_i(x) \mapsto T(\varphi_1), P_i(y) \mapsto
  T(\varphi_1)[x \mapsto y]])$$ where $T(\varphi_1)[x \mapsto y]$
denotes the formula obtained by substituting variable $y$ for all free
occurrences of $x$ in $T(\varphi_1)$, and $T(\varphi_2[P_i(x) \mapsto
  T(\varphi_1), P_i(y) \mapsto T(\varphi_1)[x \mapsto y]])$ denotes
substitution of any free occurrence of the form $P_i(x)$ in
$T(\varphi_1)$ and every occurrence of $P_i(y)$ by $T(\varphi_1)[x
  \mapsto y]$.  We let $\UTLlet$ be the extension of $\UTL$ by the
operator above, and similarly define $\TLlet$, $\FOtwoLTlet$, etc.

For $\varphi$ a temporal logic formula or an $\FOtwo$ formula with one
free variable, we denote by $L(\varphi)$ the set $\{ w \in
\Sigma^\omega : (w,0) \models \varphi\}$ of infinite words that
satisfy $\varphi$ at the initial position. The quantifier depth of an
$\FOtwo$ formula $\varphi$ is denoted $\mathit{qdp}(\varphi)$ and the
operator depth of a UTL formula $\varphi$ is denoted
$\mathit{odp}(\varphi)$.  In either case the length of the formula is
denoted $|\varphi|$.

The notion of a subformula of an $\FOtwoLTL$ formula is defined as
usual. For an $\FOtwoLTLlet$ formula $\varphi$, let $\subf(\varphi)$
denote the set of subformulas of the equivalent $\FOtwoLTL$ formula
$T(\varphi)$, where $T$ is the translation function defined above.

\begin{lem}
\label{lem:temp_closure}
Given an $\LTLlet$ formula $\varphi$, $|\subf(\varphi)|$
is linear in $|\varphi|$.
\end{lem}

\begin{proof}
Notice that if $\varphi = \olet ~P_i(x)~ \obe ~\varphi_1(x)~ \oin
~\varphi_2(x)$, then $|\subf(\varphi)| \leq |\subf(\varphi_1)|+
|\subf(\varphi_2)|$.  Then by structural induction it holds that for a
$\LTLlet$-formula $\varphi$, $\subf(\varphi)$ has size at most
$|\varphi|$.
\end{proof}

\myparagraph{B\"uchi Automata}
Our results will be obtained via transforming formulas to automata that accept  $\omega$-words.
We will be most concerned with \emph{generalised B\"uchi automata} (GBA).
A GBA $A$ is a 
tuple $(\Sigma, S, S_0,  \Delta, \cal{F})$ with alphabet $\Sigma$, set of states 
$S$, set of initial states $S_0 \subseteq S$, transition function 
$\Delta$ and set of sets of final states $\cal{F}$. The accepting condition is that for 
each $F \in \cal{F}$ there is a state $s \in F$ which is visited
infinitely often.  We can have labels either on states or on
transitions, but both models are equivalent.  For more details,
see \cite{lics1986-VW}. We will consider two important classes of
B\"uchi automata: the automaton $A$ is said to be \emph{deterministic
in the limit} if all states reachable from accepting states are
deterministic; $A$ is \emph{unambiguous} if for each state s each word
is accepted along at most one run that starts at $s$.

\myparagraph{Deterministic Parity Automata} For some model checking
problems, we will need to work with deterministic automata. In
particular, we will use deterministic parity automata. A deterministic
parity automaton $A$ is a tuple $(\Sigma, S, s_0, \Delta, Pr)$ with
alphabet $\Sigma$, set of states $S$, an initial state $s_0 \in S$,
transition function $\Delta$ and a priority function $Pr$ mapping each
state to a natural number. The transition function $\Delta$ maps each
state and symbol of the alphabet exactly to one new state. A run of
such an automaton on input $\omega$-word induces an infinite sequence
of priorities. The acceptance condition is that the highest infinitely
often occurring priority in this sequence is even.

\myparagraph{Complexity Classes} Our complexity bounds involve
\emph{counting classes}.  \#P is the class of functions $f$ for which
there is a non-deterministic polynomial-time Turing Machine $T$ such
that $f(x)$ is the number of accepting computation paths of $T$ on
input $x$.  A complete problem for \#P is \#SAT, the problem of
counting the number of satisfying assignments of a given boolean
formula.  We will be considering computations of probabilities, not
integers, so our problems will technically not be in {\#P}; but some
of them will have representations computable in the related class
$FP^{\#P}$, and will be $\#P$-hard. For brevity, we will sometimes
abuse notation by saying that such probability computation problems
are $\#P$-complete.  The class of functions \#EXP is defined
analogously to \#P, except with $T$ a non-deterministic
exponential-time machine.  We will deal with a decision version of
\#EXP, PEXP, the set of problems solvable by nondeterministic Turing
machine in exponential time, where the acceptance condition is that
more than a half of computation paths accept ~\cite{pexp}.

{\bf Notation:} In our complexity bounds, we will often write
$\mathit{poly}$ to denote a fixed but arbitrary polynomial.

\section{$\FOtwo$ model theory and succinctness} \label{sec:infrastructure}
We now discuss the model theory of $\FO2$, summarizing and slightly
extending the material presented in Etessami, Vardi, and Wilke
~\cite{fo2_utl} and in Weis and Immerman \cite{wi}.

Recall that we will consider strings over alphabet $\Sigma =
2^{\mathcal{P}}$, where $\mathcal{P}$ is the set of unary predicates
appearing in the input $\FO2[<]$ formula.  We start by recalling the
small-model property of $\FO2$ that underlies the NEXPTIME
satisfiability result of Etessami, Vardi, and Wilke~\cite{fo2_utl}, it
is also implicit in Theorem 6.2 of \cite{wi}.

The \emph{domain} of a word $u \in \Sigma^* \cup \Sigma^\omega$ is the
set $\mathrm{dom}(u) = \{ i \in \mathbb{N} : 0 \leq i < |u| \}$ of
positions in $u$.  The \emph{range} of $u$ is the set $\mathrm{ran}(u)
= \{ u_i : i \in \mathrm{dom}(u) \}$ of letters occurring in $u$.
Write also $\mathrm{inf}(u)$ for the set of letters that occur
infinitely often in~$u$.

Given a finite or infinite word
$u \in \Sigma^* \cup \Sigma^\omega$, a position
$i \in \mathrm{dom}(u)$, and $k \in \mathbb{N}$, we define
the \emph{$k$-type of $u$ at position $i$} to be the set of $\FO2[<]$
formulas
\[ \tau_k(u,i) = \{ \varphi(x) : \mathrm{qdp}(\varphi)=k \mbox{ and }
                                 (u,i) \models \varphi \} \, .\]

Given $u,v \in \Sigma^* \cup \Sigma^\omega$ and positions $i \in
\mathrm{dom}(u)$ and $j \in \mathrm{dom}(v)$, write $(u,i) \sim_k
(v,j)$ if and only if $\tau_k(u,i)=\tau_k(v,j)$.  Furthermore, we
write $u \sim_k v$ for two strings $u$, $v \in \Sigma^* \cup
\Sigma^\omega$ if for all $\FO2[<]$-formulas $\varphi(x)$ of
quantifier depth at most $k$ we have $(u,0) \models \varphi$ iff
$(v,0) \models \varphi$.

The small model property of \cite{fo2_utl} can then be stated as follows:

\begin{prop}[\cite{fo2_utl}]
\label{shortWord}
Let $\Sigma=2^{\mathcal{P}}$. Then (i)~For any string $u \in
\Sigma^*$ and positive integer $k$ there exists $v \in \Sigma^*$ such
that $u \sim_k v$ and $|v| \in 2^{O(|\mathcal{P}|k)}$; (ii)~for any
infinite string $u \in \Sigma^\omega$ and positive integer $k$ there
are finite strings $v$ and $w$, with $|v|,|w| \in
2^{O(|\mathcal{P}|k)}$, such that $u \sim_k vw^{\omega}$.
\end{prop}

For completeness, we give a constructive proof of Proposition
\ref{shortWord}, which will be used in one of our translations of
$\FO2$ to automata.  This is Lemma~\ref{lem:parity} at the end of this
section.  For this it is convenient to use the following inductive
characterisation of $\sim_k$, which is proven in \cite{fo2_utl} by a
straightforward induction:
\begin{prop}[\cite{fo2_utl}]
Let $u,v \in \Sigma^* \cup \Sigma^\omega$.  Then
$\tau_k(u,i)=\tau_k(v,j)$ if and only if (i)~$u_i=v_j$,
(ii)~$\{ \tau_{k-1}(u,i') : i' < i \}=\{ \tau_{k-1}(v,j') : j' < j \}$, and
(iii)~$\{ \tau_{k-1}(u,i') : i' > i \}=\{ \tau_{k-1}(v,j') : j' > j \}$.
\label{prop:ind-char}
\end{prop}

The next proposition states that we can collapse any two positions in
a string that have the same $k$-type without affecting the $k$-type of
the string.

\begin{prop}[\cite{fo2_utl}]
Let $u \in \Sigma^* \cup \Sigma^\omega$ and let $i < j$ be such that
$(u,i) \sim_k (u,j)$.  Writing $u = u_1\ldots u_ju'$, we have $u \sim_k
u_1\ldots u_iu'$.
\label{prop:collapse}
\end{prop}

From these two propositions it follows that every finite string is
equivalent under $\sim_k$ to a string of length exponential in $k$ and
$|\mathcal{P}|$.
\begin{prop}
Given a nonnegative integer $k$, for all strings $u \in \Sigma^*$ there
exists a string $v \in \Sigma^*$ such that $u \sim_k v$ and $|v|$ is
bounded by $2^{O(|\mathcal{P}|k)}$.
\label{prop:rep}
\end{prop}
\begin{proof}
We prove by induction on $k$ that the set $\{ \tau_k(u,i) : i \in
\mathrm{dom}(u)\}$ of $k$-types occurring along $u$ has size at most
$|\Sigma|(2|\Sigma|+2)^k$.

The base case $k=0$ is clear.

For the induction step, assume that the number of $(k-1)$-types
occurring along $u$ is at most $|\Sigma|(2|\Sigma|+2)^{k-1}$.  Define
a \emph{boundary point} in $u$ to be the position of the first or last
occurrence of a given $(k-1)$-type.  Then there are at most
$2|\Sigma|(2|\Sigma|+2)^{k-1}$ boundary points.  But by
Proposition~\ref{prop:ind-char} the $k$-type at a given position $i$
in $u$ is determined by $u_i$, the set of boundary points strictly
less than $i$, and the set of boundary points strictly greater than
$i$.  Thus the number of $k$-types along $u$ is at most
\begin{gather}
(|\Sigma|+1)2|\Sigma|(2|\Sigma|+2)^{k-1} = |\Sigma|(2|\Sigma|+2)^k \, .
\label{eq:bound10}
\end{gather}

By Proposition~\ref{prop:collapse}, given any string $v$ in which
there are two distinct positions with the same $k$-type there exists a
shorter string $w$ with $v \sim_k w$.  From the
bound~(\ref{eq:bound10}) on the number of boundary points, we conclude
that there exists a string $v$ such that $u \sim_k v$ and $|v| \leq
|\Sigma|(2|\Sigma|+2)^k \in 2^{O(|\mathcal{P}|k)}$.
\end{proof}

The relation $\sim_k$ is also easy to compute:

\begin{prop}
Given $u,v \in \Sigma^*$ of length at most $h$
we can compute whether $u \sim_k v$ in time at most 
$h2^{O(|\mathcal{P}|k)}$.
\label{prop:compute}
\end{prop}
\begin{proof}
For $m=0,1,\ldots,k$ we successively pass along $u$ labelling each
position $i$ with its $m$-type $\tau_m(u,i)$.  Each rank $m$ requires
two passes: we pass leftward through $u$ computing the set of
$(m-1)$-types to the left of each position, then we pass rightward
computing the set of $(m-1)$-types to the right of each position.
This requires $2k$ passes, with each pass taking time linear in $h$
and at most quadratic in the number of $k$-types that occur along
$u$. The bound now follows using the estimate of the number of types
given in Proposition \ref{prop:rep}.
\end{proof}

Combining Propositions~\ref{prop:rep} and~\ref{prop:compute} we get:
\begin{cor}
Given $k$ there exists a set $\mathrm{Rep}_k(\Sigma) \subseteq
\Sigma^*$ of \emph{representative strings} such that each $v \in
\mathit{Rep_k}(\Sigma)$ has $|v| \leq |\Sigma|(2|\Sigma|+2)^k$ and for
each string $u \in \Sigma^*$ there exists a unique $v \in
\mathrm{Rep}_k(\Sigma)$ such that $u \sim_k v$.  Moreover
$\mathrm{Rep}_k(\Sigma)$ can be computed from $k$ in time
$2^{2^{O(|\mathcal{P}|k)}}$.
\label{corl:rep}
\end{cor}

The following result is classical, and can be proven using games.

\begin{prop}
Given $u,v \in \Sigma^*$ and $u',v' \in \Sigma^\omega$, for all $k$ if
$u \sim_k v$ and $u' \sim_k v'$ then $uu' \sim_k vv'$.
\label{prop:comp}
\end{prop}

From the above we infer that the equivalence class of an infinite
string under $\sim_k$ is determined by a prefix of the string and the
set of letters appearing infinitely often within it.
\begin{prop}
Fix $k \in \mathbb{N}$.  Given $u=u_0u_1\ldots \in \Sigma^\omega$, there
exists $N \in \mathbb{N}$ such that for all $n \geq N$ and any 
word $w \in \Sigma^\omega$ with
$\mathrm{inf}(w)=\mathrm{ran}(w)=\mathrm{inf}(u)$ it holds that
$u \sim_k u_0u_1\ldots u_{n}w$.
\label{prop:bound}
\end{prop}
\begin{proof}
Define a strictly increasing sequence of integers $n_0 < n_1 <
\ldots < n_k$ inductively as follows.

Let $n_0$ be such that for all $i \geq n_0$ letter $u_i$ occurs
infinitely often in $u$.  For $0 < s \leq k$ let $n_{s}$ be such that
$\mathrm{ran}(u_{n_{s-1}}\ldots u_{n_s})=
\mathrm{inf}(u)$.
Now define $N:=n_k$.

Let $n \geq N$ and let $v:=u_0u_1\ldots u_{n}w$ for some $w$ such that
$\mathrm{inf}(w)=\mathrm{ran}(w)=\mathrm{inf}(u)$.
We claim that for all $0 \leq s \leq k$:
\begin{enumerate}[(1)]
\item if $i \leq n_s$ then $\tau_s(u,i)=\tau_s(v,i)$;
\item if $i,j > n_s$ then $\tau_s(u,i)=\tau_s(v,j)$ if $u_i = v_j$.
\end{enumerate}
This claim entails the proposition.  We prove the claim by induction
on $s$.  The base case $s=0$ is obvious.

The induction step for Clause 1 is as follows.  Suppose that $i \leq
n_s$; we must show that $\tau_s(u,i)=\tau_s(v,i)$.  Certainly
$u_i=v_i$ since $u$ and $v$ agree in the first $N$ letters.  Similarly
for all $j < i$ we have $\tau_{s-1}(u,j)=\tau_{s-1}(u,j)$ by Parts 1
and 2 of the induction hypothesis.  Now for all $j > i$ there exists
$j' > i$ such that $u_j=v_{j'}$ and hence by Part 2 of the induction
hypothesis $\tau_{s-1}(u,j)=\tau_{s-1}(v,j')$.  We conclude that
$\tau_s(u,i)=\tau_s(v,i)$ by Proposition~\ref{prop:ind-char}.

The induction step for Clause 2 is as follows.  Suppose that $i,j >
n_s$ and $u_i=v_j$; we must show that $\tau_s(u,i)=\tau_s(v,j)$.  We
will again use Proposition~\ref{prop:ind-char}.  Certainly for all $i'
> i$ there exists $j' > j$ such that $u_{i'}=v_{j'}$ and hence
$\tau_{s-1}(u,i')=\tau_{s-1}(v,j')$.  Now let $i' < i$.  If $i' \leq
n_{s}$ then $i' < j$, $u_{i'}=v_{i'}$ and hence
$\tau_{s-1}(u,i')=\tau_{s-1}(v,i')$.  Otherwise suppose $n_s < i' <
i$.  By definition of $n_{s}$ there exists $j'$, $n_{s-1} < j' \leq
n_s$ such that $u_{i'}=v_{j'}$.  Then
$\tau_{s-1}(u,i')=\tau_{s-1}(u,j')$ by Clause 2 of the induction
hypothesis.
\end{proof}

Combining Proposition \ref{prop:comp} and Proposition \ref{prop:bound}, we complete
the proof of Proposition \ref{shortWord}, giving a slight strengthening of the conclusion
for infinite words.

\begin{lem}
For any string $u \in \Sigma^\omega$ and positive integer $k$ there
exists $v \in \Sigma^*$ with $|v| \in 2^{O(|\mathcal{P}|k)}$ such that $v \sim_k
u'$ for infinitely many prefixes $u'$ of $u$, and $u \sim_k
vw^\omega$, where $w$ is a list of the letters occurring infinitely
often in $u$.
\label{lem:parity}
\end{lem}
\subsection{$\FO2$ and temporal logic}

We now examine the relationship between $\fotwo$ and   $\UTL$. Again we will be summarizing
previous  results while adding some new ones about the complexity of translation.

As mentioned previously, Etessami, Vardi and Wilke~\cite{fo2_utl} have studied the
expressiveness and complexity of $\fotwo$ on words.  
They show that $\fotwo$ has the same expressiveness as
unary temporal logic $\UTL$, giving a linear translation of $\UTL$ into
$\FO2$, and an exponential translation in the reverse direction.

\begin{lem}[\cite{fo2_utl}] \label{thm:fo2utl}
Every $\FOtwo$ formula $\varphi(x)$ can be converted to an equivalent
$\UTL$ formula $\varphi'$ with $|\varphi'| \in
2^{O(|\varphi|(\mathit{qdp}(\varphi)+1))}$ and $\mathit{odp}(\varphi')
\leq 2\,\mathit{qdp}(\varphi)$.  The translation 
runs in time polynomial in the size of the output.
\end{lem}

With regard to complexity,~\cite{fo2_utl} shows that satisfiability
for $\fotwo$ over finite words or $\omega$-words is NEXP-complete.
The NEXP upper bound follows immediately from their ``small model''
theorem (see Proposition \ref{shortWord} stated earlier).
NEXP-hardness is by reduction from a tiling problem.  This reduction
requires either the use of the successor predicate, or consideration
of models where an arbitrary Boolean combination of predicates can
hold, that is, they consider words over an alphabet of the form
$\Sigma = 2^{\{P_1, P_2, \ldots, P_n\}}$.

The NEXP-hardness result for $\FOtwoLT$ does not carry over from
satisfiability to model checking since the collection of alphabet
symbols that can appear in a word generated by the system being
checked is bounded by the size of the system.  However the complexity
of model checking is polynomially related to the complexity of
satisfiability when the latter is measured as a function of both
formula size and alphabet size.  Hence in the rest of the section we
will deal with words over alphabet $\Sigma = \{ P_0, P_1, \ldots, P_n
\}$, i.e., in which a unique proposition holds in each position.  We
call this the \emph{unary alphabet restriction}.

 One obvious approach to obtaining upper bounds for model checking
 $\FOtwoLT$ would be to give a polynomial translation to $\TL$, and
 use logic-to-automata translation for $\TL$.  Without the unary
 alphabet restriction an exponential blow-up in translating from
 $\FOtwoLT$ to $\TL$ was shown necessary by Etessami, Vardi, and
 Wilke:

\begin{prop}[\cite{fo2_utl}]
There is a sequence $(\psi_n)_{n \ge 1}$ of $\FOtwoLT$ sentences
over $\{P_1, P_2$, $\ldots, P_n\}$ of size ${\rm{O}}(n^2)$ such that the
shortest temporal logic formula equivalent to $\psi_n$ has size
$2^{\Omega(n)}$.
\end{prop}

The sequence given in~\cite{fo2_utl} to prove the above theorem is
$$\psi_n = \forall x ~ \, \forall y ~\,
(\bigwedge_{i<n}(P_i(x) \leftrightarrow P_i(y)) \rightarrow
(P_n(x) \leftrightarrow P_n(y))).$$ In particular, their proof does not apply
under the unary alphabet restriction.  However below we show that the
exponential blow-up is necessary even in this restricted setting.  Our
proof is indirect; it uses the following result about extensions of
$\FOtwo$ with let definitions:

\begin{lem}\label{lem:FO2long}
There is a sequence $(\varphi_n)_{n \ge 1}$ of $\FOtwoLTlet$ sentences
mentioning predicates $\{P_1, P_2, \ldots, P_n\}$ such that the shortest
model of $\varphi_n$ under the unary alphabet restriction has size
$2^{\Omega(|\varphi_n|)}$.
\end{lem}
\begin{proof}
We define $\varphi_n$ as follows.  
\begin{table}[h!]
$\begin{array}{l}
\varphi_n = 
\olet ~R_1(x) ~\obe~ \exists y \,(y \le x \wedge P_1(y))~\oin \\
\olet ~R_2(x) ~\obe~ \exists y \,(y \le x \wedge P_2(y) \wedge (R_1(x) 
\leftrightarrow R_1(y)))~\oin \\
\ldots \\
\olet ~ R_n(x) ~\obe~ \exists y\left(y \le x \wedge P_n(y) \wedge 
\displaystyle\bigwedge_{k=1}^{n-1} (R_k(x) \leftrightarrow R_k(y))\right)\\
\oin \,  \forall x \, 
\displaystyle\bigwedge_{i=1}^n 
\exists y  \, \left(
( \neg(R_i(x) \leftrightarrow R_i(y)) \wedge
\displaystyle\bigwedge_{j \neq i} (R_j(x) \leftrightarrow R_j(y))\right)
\end{array}$
\end{table}

\noindent The body of the nested sequence of let definitions states
that for all $x$ and for all $1 \leq i \leq n$ there exists $y$ such
that the vector of formulas $(R_1(x), R_2(x), \ldots, R_n(x))$ has the
same truth value as the vector $(R_1(y), R_2(y), \ldots, R_n(y))$ in
all but position $i$.  Hence the vector $(R_1(x), R_2(x), \ldots,
R_n(x))$ must take all $2^n$ possible truth values as $x$ ranges over
all positions in the word, i.e., any model of $\varphi_n$ must have
length at least $2^n$.

We now claim that $\varphi_n$ is satisfiable.  To show this,
recursively define a sequence of words $w^{(k)}$ over alphabet $\Sigma
= \{P_0,P_1,\ldots,P_n\}$ by $w^{(0)} = \varepsilon$ and $w^{(k+1)} =
w^{(k)} P_{n-k} w^{(k)}$, where $0 \leq k < n$.  Finally write $w =
w_n P_0$.  Then the vector of truth values $(R_1(x), R_2(x), \ldots,
R_n(x))$ counts down from $2^n-1$ to $0$ in binary as one moves along
$w$.
\end{proof}

In contrast, we show that basic temporal logic enhanced with let
definitions has the small model property:

\begin{lem}\label{lem:TLletshort}
There is a polynomial $poly$ such that every satisfiable $\TLlet$ formula
$\varphi$ has a model of size $poly(|\varphi|)$.
\end{lem}

\begin{proof}
In~\cite[Section 5]{fo2_utl}, Etessami, Vardi, and Wilke prove a small
model property for $\TL$, which follows the same lines as the one
given for $\FO2$, but with polynomial rather than exponential bounds
on sizes. Instead of using types based on quantifier-rank, the notion
of type is based on the nesting of modalities; they thus look at modal
$k$-type, where $k$ is the nesting of modalities in $\varphi$.  It was
shown how to collapse infinite $\omega$-words in order to get
"smaller" $\omega$-words with essentially the same type
structure. Then in Lemma $4$ of \cite{fo2_utl} it is shown that for
each $u \in \Sigma^\omega$ there are strings $v$, $w$ such that the
type of $u$ at position $0$ is equal to the type of $vw^\omega$ at
position $0$ and the length of both $v$ and $w$ is less than
$(t+1)^2$, where $t$ is number of types occurring along $u$ (that is,
a polynomial version to Proposition \ref{shortWord}).
 
 The type is determined by the predicate and the combination of
 temporal subformulas of $\varphi$ holding at the given position. Each
 temporal subformula, i.e. subformula which starts with
 $\DIAMOND$~or~$\DIAMONDM$, can change its truth value at most once
 along the infinite word. Therefore there are at most polynomially
 many (in $|\Sigma|$ and in number of temporal subformulas of
 $\varphi$) different combinations and so also types along $u$.

Lemma \ref{lem:temp_closure} tells us that number of temporal subformulas of $\varphi$ is 
linear in $|\varphi|$, and therefore the number of types
$t$ occurring along any word is polynomial in $|\varphi|$.  Thus
applying the above-mentioned type-collapsing argument of~\cite{fo2_utl}
we conclude that there is a polynomial size model of $\varphi$.
\end{proof}
 
The small model property for $\TLlet$ will allow the lifting of NP
model-checking results to this language. Most relevant to our
discussion of succinctness, it can be combined with the previous
result to show that $\FOtwoLT$ is succinct with respect to $\TL$:

\begin{prop}\label{thm:nofo2tl}
Even assuming the unary alphabet restriction, there is no polynomial
translation from $\FOtwoLT$ formulas to equivalent $\TL$-formulas.
\end{prop}
\begin{proof}
Proof by contradiction. Assuming there were such a polynomial
translation, we could apply it locally to the body of every let
definition in an $\FOtwoLTlet$ formula.  This would allow us to
translate an $\FOtwoLTlet$ formula to a $\TLlet$ formula of polynomial
size.  Therefore it would follow from Lemma \ref{lem:TLletshort} that
every $\FOtwoLTlet$ formula that is satisfiable has a polynomial sized
model, which is a~contradiction of Lemma \ref{lem:FO2long}.
\end{proof}

Proposition~\ref{thm:nofo2tl} shows that we cannot obtain better
bounds for $\FOtwoLT$ merely by translation to $\TL$.  Weis
\cite{WeisPhd} showed an NP-bound on satisfiability of $\FOtwoLT$
under the unary alphabet restriction (compared to NEXP-completeness of
satisfiability in the general case).  His approach is to show that
models realise only polynomially many types.  We will later show that
the approach of Weis can be extended to obtain complexity bounds for
model checking $\FOtwoLT$ that are as low as one could hope, i.e.,
that match the complexity bounds for the simplest temporal logic,
$\TL$.  We do so by building sufficiently small unambiguous B\"uchi
automata for $\FOtwoLT$ formulas.

\section{Translations}
\label{sec:trans}

This section contains a key contribution of this paper---three
logic-to-automata translations for $\UTL$, $\FO2$, and $\FOtwoLT$.  We
will later use these translations to obtain upper complexity bounds
for model checking both non-deterministic and probabilistic
systems. As we will show, for most of the problems it is sufficient to
translate a given formula to an unambiguous B\"uchi automaton.  Our
first translation produces such an automaton from a given $\UTL$
formula.  This is then lifted to full $\FO2$ via a standard syntactic
transformation from $\FO2$ to $\UTL$.  Our second translation goes
directly from the stuffer-free fragment $\FOtwoLT$ to unambiguous
B\"{u}chi automata, and is used to obtain optimal bounds for this
fragment.  Our third translation constructs a deterministic parity
automaton from an $\FO2$ formula.  Having a deterministic automaton is
necessary for solving two-player games and quantitative model checking
of Markov decision processes.

\subsection{Translation I: From UTL to unambiguous B\"uchi automata}
\label{subsec:utl2ba}
We begin with a translation that takes $\UTL$ formulas to B\"{u}chi
automata.  Combining this with the standard syntactic transformation
of $\FO2$ to $\UTL$, we obtain a translation from $\FO2$ to B\"{u}chi
automata.

Recall from the preliminaries section that a B\"{u}chi automaton $A$
is said to be \emph{deterministic in the limit} if all accepting
states and their descendants are deterministic, and that $A$ is
\emph{unambiguous} if each word has at most one accepting run.

We will aim at the following result:
\begin{thm}
\label{BW:detLimit}
Let $\varphi$ be a {UTL} formula over set of propositions $\mathcal{P}$ with
operator depth $n$ with respect to $\CIRC$ and $\CIRCM$.  Given an
alphabet $\Sigma\subseteq 2^{\mathcal{P}}$, there is a family of at most
$2^{|\varphi|^2}$ B\"{u}chi automata $\{ A_i\}_{i \in I}$ such that
(i)~$\{ w \in \Sigma^\omega : w \models \varphi\}$ is the disjoint
union of the languages $L(A_i)$; (ii)~$A_i$ has at most
$O(|\varphi||\Sigma|^{n+1})$ states; (iii)~$A_i$ is unambiguous and
deterministic in the limit; (iv) there is a polynomial-time procedure
that outputs $A_i$ given input $\varphi$ and index $i \in I$.
\end{thm}

We first outline the construction of the family $\{A_i\}$.  Let
$\varphi$ be a formula of $\TL$ over set of atomic
propositions~$\mathcal{P}$.  Following Wolper's construction
\cite{wolper}, define $\mathit{cl}(\varphi)$, the \emph{closure} of
$\varphi$, to consist of all subformulas of $\varphi$ (including
$\varphi$) and their negations, where we identify $\neg\neg\psi$ with
$\psi$.  Furthermore, say that $\vs \subseteq \mathit{cl}(\varphi)$ is
a \emph{subformula type} if (i)~for each formula
$\psi\in\mathit{cl}(\varphi)$ precisely one of $\psi$ and $\neg\psi$
is a member of $\vs$; (ii)~$\psi \in \vs$ implies $\DIAMOND \psi,
\DIAMONDM \psi \in \vs$; (iii)~$\psi_1 \wedge \psi_2 \in \vs$ iff
$\psi_1 \in \vs$ and $\psi_2 \in \vs$.  Given subformula types $\vs$
and $\vt$, write $\vs \sim \vt$ if $\vs$ and $\vt$ agree on all
formulas whose outermost connective is a temporal operator, i.e., for
all formulas $\psi$ we have $\DIAMOND\psi\in\vs$ iff
$\DIAMOND\psi\in\vt$, and $\DIAMONDM\psi\in\vs$ iff
$\DIAMONDM\psi\in\vt$.  Note that these types are different from the
types based on modal depth considered before.

Fix an alphabet $\Sigma \subseteq 2^{\mathcal{P}}$ and write
$\mathit{tp}^\Sigma_\varphi$ for the set of subformula types $\vs
\subseteq \mathit{cl}(\varphi)$ with $\vs \cap P \in \Sigma$.  In
subsequent applications $\Sigma$ will arise as the set of
propositional labels in a structure to be model checked.
Following~\cite{wolper} we define a generalised B\"{u}chi automaton
$A^\Sigma_{\varphi} = (\Sigma,S,S_0,\Delta,\lambda,\mathcal{F})$ such
that $L(A^\Sigma_{\varphi}) = \{ w \in \Sigma^\omega : (w,0) \models
\varphi\}$.  The set of states is $S=\mathit{tp}^\Sigma_\varphi$, with
the set $S_0$ of initial states comprising those $\vs \in
\mathit{tp}^\Sigma_\varphi$ such that (i)~$\varphi \in \vs$ and
(ii)~$\DIAMONDM\psi \in\vs$ if and only if $\psi\in \vs$ for any
formula $\psi$.  The state labelling function $\lambda : S \rightarrow
\Sigma$ is defined by $\lambda(\vs) = \vs \cap P$.  The transition
relation $\Delta$ consists of those pairs $(\vs,\vt)$ such that
\begin{enumerate}[(i)]
\item[(i)]
$\DIAMONDM \psi \in \vt$ iff either
$\psi\in\vt$ or $\DIAMONDM\psi \in \vs$;
\item[(ii)]
$\DIAMOND\psi\in\vs$ and
$\psi\not\in\vs$ implies
$\DIAMOND\psi \in \vt$;
\item[(iii)]
$\neg\DIAMOND\psi\in\vs$ implies
$\neg\DIAMOND\psi\in\vt$.
\end{enumerate}
The collection of accepting sets is $\mathcal{F} = \{
F_{\DIAMOND\psi} : \DIAMOND\psi \in
\mathit{cl}(\varphi) \}$, where $F_{\DIAMOND\psi} = \{
\vs : \psi\in\vs \mbox{ or }
\DIAMOND\psi\not\in\vs \}$.

A run of $A^\Sigma_{\varphi}$ on a word $u \in \Sigma^\omega$ yields a
function $f : \mathbb{N} \rightarrow 2^{\mathit{cl}(\varphi)}$.
Moreover it can be shown that if the run is accepting then for all
formulas $\psi \in \mathit{cl}(\varphi)$, $\psi \in f(i) \Rightarrow
(u,i) \models \psi$~\cite[Lemma 2]{wolper}.  But since $f(i)$ contains
each subformula or its negation, we have $\psi \in f(i)$ if and only
if $(u,i) \models \psi$ for all $\psi \in \mathit{cl}(\varphi)$.  We
conclude that $A^\Sigma_{\varphi}$ is unambiguous and accepts the
language $L(\varphi)$.  The following lemma summarises some structural
properties of the automaton $A^\Sigma_{\varphi}$.

\begin{lem}
Consider the automaton $A^\Sigma_{\varphi}$ as a directed graph with
set of vertices $S$ and set of edges $\Delta$.  Then (i)~states $\vs$
and $\vt$ are in the same strongly connected component iff $\vs \sim
\vt$; (ii)~each strongly connected component has size at most
$|\Sigma|$; (iii)~the dag of strongly connected components has depth
at most $|\varphi|$ and outdegree at most $2^{|\varphi|}$; (iv)
$A^\Sigma_{\varphi}$ is deterministic within each strongly connected
component, i.e., given transitions $(\vs,\vt)$ and $(\vs,\vu)$ with
$\vs,\vt$ and $\vu$ in the same strongly connected component, we have
$\vt=\vu$ if and only if $\lambda(\vt)=\lambda(\vu)$.
\label{lem:observations}
\end{lem}
\begin{proof}
(i)
If $\vs \sim \vt$ then by definition of the transition relation
$\Delta$ we have that $(\vs,\vt) \in \Delta$.  Thus $\vs$ and $\vt$
are in the same connected component.  Conversely, suppose that $\vs$
and $\vt$ are in the same connected component.  By clauses (i) and
(iii) in the definition of the transition relation $\Delta$ we have
that $\DIAMONDM \psi \in \vs$ iff $\DIAMONDM \psi \in \vt$ and
likewise $\neg\DIAMOND \psi \in \vs$ iff $\neg\DIAMOND \psi \in \vt$.
But for each formula $\psi \in
\mathit{cl}(\varphi)$ either $\vs$ contains $\psi$ or its negation,
and similarly for $\vt$; it follows that $\vs \sim \vt$.

(ii) If $\vs \sim \vt$, then $\vs = \vt$ if and only if
  $\lambda(\vs)=\lambda(\vt)$.  Thus the number of states in an SCC is at
  most the number $|\Sigma|$ of labels.

(iii) Suppose that $(\vs,\vt) \in \Delta$ is an edge connecting two
  distinct SCC's, i.e., $\vs \not\sim \vt$.  Then there is a
  subformula $\DIAMOND \psi \in \vs$ such that
  $\neg\DIAMOND \psi \in \vt$.  Note that $\neg\DIAMOND \psi$ lies in
  all states reachable from $\vt$ under $\Delta$.  Since there at most
  $|\varphi|$ such subformulas, we conclude that the depth of the DAG
  of SCC's is at most $|\varphi|$.

(iv) This follows immediately from (i).
\end{proof}

We proceed to the proof of Theorem~\ref{BW:detLimit}.
\begin{proof}
We first treat the case $n=0$, i.e., $\varphi$ does not mention
$\CIRC$ or $\CIRCM$.

Let $A^\Sigma_{\varphi} = (\Sigma,S,S_0,\Delta, \lambda, \mathcal{F})$
be the automaton corresponding to $\varphi$, as defined above.  For
each path $\pi=C_0,C_1,\ldots,C_k$ of SCC's in the SCC dag of
$A^\Sigma_{\varphi}$ we define a sub-automaton $A_\pi$ as follows.
$A_\pi$ has set of states $S_\pi = C_0 \cup C_1 \cup \cdots \cup C_k$;
its set of initial states is $S_0 \cap S_\pi$; its transition relation
is $\Delta_\pi = \Delta \cap (S_\pi \times S_\pi)$, i.e., the
transition relation of $A^\Sigma_{\varphi}$ restricted to $S_\pi$; its
collection of accepting states is $\mathcal{F}_\pi = \{ F \cap C_k : F
\in \mathcal{F} \}$.

It follows from observations (ii) and (iii) in
Lemma~\ref{lem:observations} that $A_\pi$ has at most
$|\varphi||\Sigma|$ states, and from observation (iii) that there are
at most $2^{|\varphi|^2}$ such automata.  Since $A^\Sigma_{\varphi}$
is unambiguous, each accepting run of $A^\Sigma_{\varphi}$ yields an
accepting run of $A_\pi$ for a unique path $\pi$, and so the
$L(A_\pi)$ partition $L(A^\Sigma_{\varphi})$.

Finally, $A_\pi$ is deterministic in the limit since all accepting
states lie in a bottom strongly connected component, and all states in
such a component are deterministic by 
Lemma~\ref{lem:observations}(iv).  If we convert $A_\pi$ from a
generalised B\"{u}chi automaton to an equivalent B\"{u}chi automaton
(using the construction from~\cite{wolper}), then the resulting
automaton remains unambiguous and deterministic in the limit.  This
transformation touches only the bottom strongly connected component of
$A_\pi$, whose size will become at most quadratic.

This completes the proof in case $n=0$.  The general case can be
handled by reduction to this case.  A {UTL} formula $\varphi$ can be
transformed to a normal form such that all next-time $\CIRC$ and
last-time $\CIRCM$ operators are pushed inside the other Boolean and
temporal operators.  Now the formula can be regarded as a $\TL$
formula $\varphi'$ over an extended set of propositions $\{ \CIRC^i P,
\CIRCM^i P : 0 \leq i \leq n, P \in \mathcal{P}\}$.  Applying the case
$n=0$ to $\varphi'$ we obtain a family of automata $\{A'_i\}$ over
alphabet $\Sigma' = 2^{\mathcal{P}'}$ such that
$L(A^{\Sigma'}_{\varphi'}) = \bigcup_i L(A'_i)$, $A'_i$ is unambiguous
and deterministic in the limit, and $A'_i$ has at most
$O(|\varphi'||\Sigma'|)=O(|\varphi||\Sigma|^n)$ states.

Now we can construct a deterministic transducer $T$ with $|\Sigma|^n$
states that transforms (in the natural way) an $\omega$-word over
alphabet $\Sigma$ into an $\omega$-word over alphabet $\Sigma'$.  Such
a machine can be made deterministic by having $T$ produce its
output $n$ positions behind the input.  To do this we maintain an
$n$-place buffer in the states of $T$, which requires $|\Sigma|^n$
states.

We construct automaton $A_i$ over alphabet $\Sigma$ by composing
$A_i'$ with $T$, i.e., by synchronising the output of $T$ with the
input of $A'_i$.  The number of states of the composition is the
product of the number of states of $A'_i$ and $T$ which are consistent
with respect to their label in $\Sigma'$.  Thus the product has at most
$O(|\varphi||\Sigma|^{n+1})$ states.

This completes the proof of Theorem~\ref{BW:detLimit}.
\end{proof}

From Theorem~\ref{BW:detLimit} we can get a translation of $\FOtwo$
to automata with bounds as stated below:

\begin{thm}
\label{thm:FO2_aut}
Given an $\FOtwo$ formula $\varphi$, there is a collection of
$2^{2^{\mathit{poly}(|\varphi|)}}$ generalised B\"uchi automata $A_i$,
each of size at most $2^{\mathit{poly}(|\varphi|)}$ such that the
languages they accept partition the language $\{w\in \Sigma^\omega : w
\models \varphi \}$. Moreover, each automaton $A_i$ is unambiguous and
can be constructed by a non-deterministic Turing machine in polynomial
time in its size.
\end{thm}
\begin{proof}
First we apply Lemma \ref{thm:fo2utl} to translate the $\FOtwo$
formula $\varphi$ to an equivalent UTL formula $\varphi'$.  We then
apply Theorem \ref{BW:detLimit} to $\varphi'$, noting that the size of
$\varphi'$ is exponential in the size of $\varphi$, while the operator
depth of $\varphi'$ is polynomial in the quantifier depth of
$\varphi$. Finally, we apply Theorem \ref{BW:detLimit} to $\varphi'$.
\end{proof}

\subsection{Translation II: From $\FOtwoLT$ to unambiguous B\"uchi automata}
The previous translation via $\UTL$ will be useful for giving bounds
on verifying both $\UTL$ and $\FO2$.  However it does not give insight
into the sublanguage $\FOtwoLT$.  We will thus give another
translation specific to this fragment.  The main idea for getting upper
bounds on verification problems for $\FOtwoLT$ will be to show that
for any $\FOtwoLT$ formula $\varphi$, the number of one-variable
subformula types realised along a finite or infinite word is
polynomial in the size of $\varphi$. Informally these subformula types
are the collections of one-variable subformulas of $\varphi$ that
might hold at a given position. Note that the types we consider here
are collections of $\FOtwoLT$ formulas, not temporal logic formulas as
in the last section. Also note the contrast with the $k$-types of
Proposition \ref{shortWord}, which consider all formulas of a given
quantifier rank.

Recall that the \emph{domain} of a word $u \in
\Sigma^* \cup \Sigma^\omega$ is the set $\mathrm{dom}(u) = \{ i \in
\mathbb{N} : 0 \leq i < |u| \}$ of positions in $u$.  
Given an $\FOtwoLT$-formula $\varphi$,
let $\mathrm{cl}(\varphi)$ denote the set of all subformulas of
$\varphi$ with at most one free variable (including atomic predicates).  Given a finite or infinite word
$u \in \Sigma^* \cup \Sigma^\omega$, a position
$i \in \mathrm{dom}(u)$, we define
the \emph{subformula type of $u$ at position $i$} to be the set of $\FOtwoLT$
formulas
\[ \tau(u,i) = \{ \psi : \psi \in \mathrm{cl}(\varphi) \mbox{ and }
                                 (u,i) \models \psi \} \, .\] We have omitted $\varphi$ in this notation 
since it will be fixed for the remainder of the proof.

\myparagraph{Few Types Property for $\FOtwoLT$}
We will base our result on the following theorem of
Weis~\cite{WeisPhd}, showing that $\FOtwoLT$ formulas divide a finite
word into a small number of segments based on subformula type:

\begin{prop}[\cite{WeisPhd}] 
\label{thm:fewtypesfin}
Let $\varphi$ be an $\FOtwoLT$-formula. A string $u \in \Sigma^*$ can
be written $u=v_1\ldots v_n$, where $v_i \in \Sigma^*$, $n$ is
polynomial in $|\varphi|$ and $|\Sigma|$, and for any two positions
$i, j$ lying in the same factor $v_k$ having the same symbol,
$\tau(u,i) = \tau(u,j)$.
\end{prop}

We will need an extension of this result to infinite words:
\begin{prop}
\label{thm:fewtypesinf}
Let $\varphi$ be an $\FOtwoLT$-formula. A string $u \in \Sigma^\omega$
can be written $u=v_1\ldots v_n$, where $v_k \in \Sigma^*$ for $k<n$
and $v_n \in \Sigma^\omega$, $n$ is
polynomial in $|\varphi|$ and $|\Sigma|$, and for any two positions $i, j$ lying
within the same factor and having the same symbol we have $\tau(u,i)
= \tau(u,j)$.
\end{prop}

\begin{proof}
We note that for any $u \in \Sigma^\omega$, from some position
onwards, the subformula type is determined only by the current
symbol. In fact, the proof of Proposition \ref{prop:bound} shows that
we have $u = vw$ for some prefix $v \in \Sigma^*$ of $u$ and
$w \in \Sigma^\omega$ such that for any two positions $i, j$ of $vw$
such that $i,j>|v|$ having the same symbol $\tau(vw,i)=\tau(vw,j)$.

Given an infinite $u$, we can take a finite prefix $v$ as above and
apply Proposition~\ref{thm:fewtypesfin} to it, adding on the infinite
interval $w$ as one additional member of the partition. Now if $i$ and
$j$ are in the final partition, then agreement on the same symbol
determines the entire set of formulas, and hence we are done.
Otherwise, fix any two positions $i, j \le |v|$ in $u$ with the symbol
$a \in \Sigma$ holding at both $i$ and $j$, and lying in the same
partition within $v$. We claim that the subformula types $\tau(u,i)$
and $\tau(u,j)$ contain the same set of formulas.  An atomic predicate
$\psi \in \rm{cl(\varphi)}$ holds at position $i$ iff it holds at $j$
by assumption, since there is only one symbol true at each position.
Positions $i$ and $j$ then by assumption satisfy the same subformula
type within $v$. But using the hypothesis on $v$ we can easily see
inductively that a subformula holds on a position within $v$ iff it
holds at that position within $vw$.
\end{proof}

We now present a result showing that the few subformula types property 
can be used to get  a better translation
to automata:

\begin{thm}
Assume the unary alphabet restriction. Then
given an $\FOtwoLT$ formula $\varphi$, there is a collection of 
$2^{\mathit{poly}(|\varphi|,|\Sigma|)}$ generalised B\"uchi automata $A_i$ 
(each of polynomial size in $|\varphi|$ and $|\Sigma|$) such that
the languages they accept are disjoint and the union of these languages is 
exactly  $\{w\in \Sigma^\omega : w \models \varphi \}$. 
Moreover, each automaton $A_i$ is unambiguous and deterministic in the limit 
and can be constructed by a non-deterministic Turing machine in polynomial-time.
\label{FO2automaton}
\end{thm}
\begin{proof}
We say that $\tau \subseteq \mathit{cl}(\varphi)$ is a
\emph{subformula pre-type} if: (i)~if $\varphi_1 \wedge \varphi_2 \in
\mathit{cl}(\varphi)$, then $\varphi_1 \wedge \varphi_2 \in \tau$ iff
$\varphi_1\in\tau$ and $\varphi_2 \in \tau$; (ii)~if $\varphi_1 \vee
\varphi_2 \in \mathit{cl}(\varphi)$, then $\varphi_1 \vee \varphi_2
\in \tau$ iff $\varphi_1\in\tau$ or $\varphi_2 \in \tau$; (iii)~if
$\neg\psi \in \mathit{cl}(\varphi)$, then $\neg\psi \in \tau$ iff
$\psi\not\in \tau$.

This notion is similar to the notion of ``subformula type of a node''
used in the prior results, except that a collection of formulas
satisfying the above property may not be consistent, since the
semantics of existential quantifiers is not taken into account.

In general the formulas in a (subformula) pre-type
$\tau$ can have either $x$ or $y$ as free variables.  We write
$\tau(x)$ for the subformula pre-type obtained by interchanging $x$
and $y$ in all formulas in $\tau$ with $y$ as free variable.  Thus all
formulas in $\tau(x)$ have free variable $x$.  We similarly define
$\tau(y)$.

An \emph{order formula} is an atomic formula
\[ \order ::= x<y \, \mid \, y<x \, \mid \, x=y \, . \]
Given $m,n \in \mathbb{N}$ let $\order_{m,n}$ denote the unique
order formula satisfied by the valuation $x,y \mapsto m,n$.

Given a pair of pre-types $\tau_1, \tau_2$, an order formula $\order$,
and a subformula $\theta$ of $\varphi$, we write
$\tau_1(x),\tau_2(y),\order \models \theta(x,y)$ to denote that when
$\theta$ is transformed by replacing top-level subformulas by their
truth values as specified by $\tau_1(x)$, $\tau_2(y)$, or $\order$,
then the resulting Boolean combination evaluates to true.  Note that
this implies that if word $w$ and positions $i,j$ satisfy $\tau_1(x)
\cup \tau_2(y) \cup \{ \order \}$, then they also satisfy $\theta$.

A \emph{closure labelling} is a function $f : \mathbb{N} \rightarrow
2^{\mathit{cl(\varphi)}}$ such that \begin{enumerate}[(1)]
\item $f(n)$ is a  pre-type for each
$n\in \mathbb{N}$ and \item 
for each $n\in \mathbb{N}$, if $\exists
y\,\theta \in \mathit{cl}(\varphi)$ then $\exists y\theta \in f(n)$
iff there exists $m \in \mathbb{N}$ such that
$f(n)(x),f(m)(y),\order_{n,m} \models \theta$.
\end{enumerate}

\noindent It is easy to see that an $\omega$-word $w:\mathbb{N}\rightarrow
\Sigma$ has a unique extension to a closure labelling $f : \mathbb{N}
\rightarrow 2^{\mathit{cl}(\varphi)}$.  Namely, $f$ is defined by
$f(n)=\{ \psi \in \mathit{cl}(\varphi) : w,n\models \psi\}$.

We now define a generalised B\"{u}chi automaton $A_{\varphi}$
corresponding to $\varphi$.  
\begin{defi}
The alphabet of $A_{\varphi}$ is $\Sigma$, and the other components of
$A_{\varphi}$ are as follows:

\emph{States.}  The states of $A_{\varphi}$ are tuples
$(\boldsymbol{s},\tau,\boldsymbol{t})$, where $\tau \subseteq
\mathit{cl}(\varphi)$ is a pre-type and $\boldsymbol{s},\boldsymbol{t}
\subseteq 2^{\mathit{cl}(\varphi)}$ are sets of pre-types of size at
most $p(|\varphi|,|\Sigma|)$, where $p$ is the polynomial from
Proposition~\ref{thm:fewtypesinf}, such that the following
\emph{consistency condition} holds: for each formula $\exists y\theta
\in \tau$ we have that either $\tau(x),\tau(y),x=y \models \theta$,
$\tau(x),\tau'(y),x<y \models \theta$ for some $\tau'\in
\boldsymbol{t}$, or $\tau'(y),\tau(x),y<x \models \theta$ for some
$\tau'\in \boldsymbol{s}$.  (This condition corresponds to the second
clause in the definition of closure labelling.)  Informally, a state
consists of an assertion about the subformula pre-types seen in the
past, the current subformula pre-type, and the subformula pre-types to
be seen in the future.

\emph{Initial State.}  A state $(\boldsymbol{s},\tau,\boldsymbol{t})$
is initial if $\boldsymbol{s}=\emptyset$ and $\varphi \in \tau$.

\emph{Accepting States.}
There is a set of accepting states $F_\tau$ for each  pre-type $\tau$.
We have $(\boldsymbol{s},\tau',\boldsymbol{t}) \in F_\tau$ if and only
if $\tau=\tau'$ or $\tau\not\in \boldsymbol{t}$.  

\emph{Transitions.}  For each $a \in \Sigma$ there is an
$a$-labelled transition from
$(\boldsymbol{s},\tau,\boldsymbol{t})$ to
$(\boldsymbol{s}',\tau',\boldsymbol{t}')$ iff (i) for the unique
proposition $P_i(x)$ in $\tau$, $P_i=a$;
(ii)~$\boldsymbol{s}'=\boldsymbol{s}\cup\{\tau\}$; (iii)~$\tau'\in
\boldsymbol{t}$; (iv)~either $\boldsymbol{t}'=\boldsymbol{t}$ or
$\boldsymbol{t}'=\boldsymbol{t}\setminus \{\tau'\}$.
\end{defi}

The following proposition, whose proof follows straightforwardly from
Proposition \ref{thm:fewtypesinf}, shows that the automaton captures
the formula:

\begin{prop}
If $(\boldsymbol{s}_0,\tau_0,\boldsymbol{t}_0),
(\boldsymbol{s}_1,\tau_1,\boldsymbol{t}_1),
(\boldsymbol{s}_2,\tau_2,\boldsymbol{t}_2),\ldots$ is an accepting run
of $A_{\varphi}$, then the function $f : \mathbb{N} \rightarrow
2^{\mathit{cl(\varphi)}}$ defined by $f(n)=\tau_n$ is a closure
labelling.  Moreover every closure labelling $f$ such that $\varphi\in
f(0)$ arises from a run of $A_{\varphi}$ in this manner.
\end{prop}

We now analyze the automaton $A_{\varphi}$. Because of the polynomial
restriction on the number of pre-types, the automaton has at most
exponentially many states. But by Proposition~\ref{thm:fewtypesinf},
any accepting run goes through only polynomially many states. For
every path $\pi$ in the DAG of strongly-connected components, we take
the subautomaton $A_\pi$ of $A_{\varphi}$ obtained by restricting to
the components in this path.  We claim that this is the required
decomposition of $A_{\varphi}$.  Note that an NP machine can construct
these restrictions by iteratively making choices of successor
components that are strictly lower in the DAG.  Clearly the automata
corresponding to distinct paths accept disjoint languages, since they
correspond to different collections of pre-types holding in the word.
One can show that for any word satisfying the formula, the unique
accepting run is the one in which the state at a position corresponds
to the pre-types seen before the positions, the pre-type seen at the
position, and the pre-types seen after the position. In particular,
this shows that each automaton is unambiguous.  Finally, because the
only nondeterministic choice is whether to leave an SCC or not, upon
reaching the bottom SCC the automaton is deterministic---hence each
automaton is deterministic in the limit.  Thus this decomposition
witnesses Theorem \ref{FO2automaton}.
\end{proof}

The above translation of $\FOtwoLT$ formulas to unambiguous B\"uchi
automata can be extended to handle formulas with successor, i.e., the
full logic $\FOtwo$, at the same time removing the unary alphabet
restriction.  Given an $\FOtwo$ formula $\varphi$ over set of
predicates $\mathcal{P}$, we can consider an ``equivalent'' $\FOtwoLT$
formula $\varphi'$ over a set of new predicates
$2^{|\varphi||\mathcal{P}|}$.  Intuitively each predicate in
$\mathcal{P}'$ specifies the truth values of all predicates in
$\mathcal{P}$ in a neighbourhood of radius $|\varphi|$ around the
current position.  Applying Theorem~\ref{FO2automaton} to $\varphi'$
we obtain a collection of double-exponentially many automata $A_i$,
each of size exponential size in $\varphi$ and $\Sigma$.  Thus, we get
a weaker version of Theorem \ref{thm:FO2_aut} of the previous
subsection, in which the size bound on the component automata has an
exponential dependence on the alphabet as well as the formula size.

\subsection{Translation III: From $\FOtwo$ to deterministic parity automata}
\label{subsec:trans}
While the previous translations are useful for relating $\FO2$ to
unambiguous automata, for some problems it is useful to have
deterministic automata.  We now give a translation of $\FO2$ formulas
to ``small'' deterministic parity automata.  We give the translation
first for the fragment $\FO2[<]$ without successor and show later how
to handle the full logic.
Specifically, we will show:
\begin{thm}
\label{DRA}
Given an $\FO2 [<]$ formula $\varphi$ over set of predicates ${\mathcal{P}}$ with
quantifier depth $k$, there exists a deterministic parity automaton
$\mathcal{A}_\varphi$ accepting the language $L(\varphi)$ such that
$\mathcal{A}_\varphi$ has $2^{2^{O(|\mathcal{P}|k)}}$ states, $2^{O(|\mathcal{P}|)}$
priorities, and can be computed from $\varphi$ in time
$|\varphi|^{O(1)} \cdot 2^{2^{O(|\mathcal{P}|k)}}$.
\end{thm}

The definition of the automaton $\mathcal{A}_{\varphi}$ in
Theorem~\ref{DRA} relies on the small-model property, as stated in
Proposition ~\ref{shortWord}.  By Lemma~\ref{lem:parity}, to know
whether $u \in \Sigma^\omega$ satisfies an $\FOtwoLT$-formula of
quantifier depth $k$ it suffices to know some $k$-type such that
infinitely many prefixes of $u$ have that type, as well as which
letters occur infinitely often in $u$. We will translate $\varphi$ to
a deterministic parity automaton $\mathcal{A}_\varphi$ that detects
this information. As $\mathcal{A}_\varphi$ reads an input string $u$
it stores a representative of the $k$-type of the prefix read so far.
By Proposition~\ref{shortWord}(i) the number of such representatives
is bounded by $2^{2^{O(|\mathcal{P}|k)}}$.  Applying
Lemma~\ref{lem:parity}, we use a parity acceptance condition to
determine whether $u$ satisfies $\varphi$, based on which
representatives and input letters occur infinitely often.

We are now ready to formally define $\mathcal{A}_\varphi$.  To this
end, define the \emph{last appearance record} of a finite string $u =
u_0\ldots u_n \in \Sigma^*$ to be the substring $\mathrm{LAR}(u) :=
u_{i_1}u_{i_2}\ldots u_{i_m}$ such that for all $k \in
\mathrm{dom}(u)$ there exists a unique $i_j \geq k$ such that
$u_{i_j}=u_k$.  Thus we obtain $\mathrm{LAR}(u)$ from $u$ by keeping
only the last occurrence of each symbol from $u$.  Write
$\mathrm{LAR}(\Sigma)$ for the set $\{ \mathrm{LAR}(u) : u \in
\Sigma^*\}$ of all possible last appearance records.  Recall also the
set of strings $\mathrm{Rep}_k(\Sigma)$ from Corollary~\ref{corl:rep}
that represent the different $k$-types of strings in $\Sigma^*$.

\begin{defi}
Let $\varphi(x)$ be an $\FO2[<]$-formula of quantifier depth $k$.  We
define a deterministic parity automaton $\mathcal{A}_\varphi$ as follows.

\begin{iteMize}{$\bullet$}
\item $\mathcal{A}_\varphi$ has set of states $\mathrm{Rep}_k(\Sigma)
  \times \mathrm{LAR}(\Sigma) \times \{0,1,\ldots,|\Sigma|\}$.
\item The initial state is $(\varepsilon,\varepsilon,0)$.
\item The transition function maps a state $(s,\ell,i)$, where
  $\ell=\ell_1\ldots \ell_j$, and input letter $a \in \Sigma$ to the
  unique state $(t,\ell',j')$ such that $sa \sim_k t$, $\ell' =
  \mathrm{LAR}(\ell a)$, $j'=0$ if $a$ does not occur in $\ell$ and
  otherwise $\ell_{j'}=a$.
\item The set of priorities is $0,1,\ldots,2|\Sigma|+1$.
\item The priority of state $(s,\ell,i)$ where $\ell=\ell_1\ell_2\ldots
  \ell_j$ is given by
\[ pr(s,\ell,i) = \left\{ \begin{array}{ll}
              2 \cdot |\ell_i\ldots \ell_j| 
              & \mbox{if $(s(\ell_i\ldots\ell_j)^\omega,0) \models \varphi$}\\
              2 \cdot |\ell_i\ldots \ell_j|+1  & \mbox{otherwise.}
\end{array} \right . \]
\end{iteMize}\smallskip
\end{defi}

\noindent It follows from Proposition~\ref{prop:comp} that in a run of
$\mathcal{A}_\varphi$ on a finite word $u =u_0u_1 \ldots
u_n \in \Sigma^*$ the last state $(s,\ell,i)$ is such that $s$ has the
same $k$-type as $u$.  Also we note that $\ell$ is the LAR of $u$ and
$i$ is the position in the previous LAR of $u_n$.

The following two results prove Theorem~\ref{DRA}:

\begin{prop}
$L(\mathcal{A}_\varphi) = \{ u \in \Sigma^\omega : (u,0) \models \varphi\}$.
\end{prop}
\begin{proof}
Let $u \in \Sigma^\omega$ and let $N$ be as in
Proposition~\ref{prop:bound}.  Suppose that the highest infinitely
often occurring priority in a run of $\mathcal{A}_{\varphi}$ on $u$ is
even.  Then there exists $n \geq N$ such that $\mathcal{A}_\varphi$ is
in state $(s,\ell,i)$ after reading $u_0u_1\ldots u_n$, where 
$\ell=\ell_1\ell_2\ldots \ell_j$, $\{\ell_i,\ldots,\ell_j\} =
\mathrm{inf}(u)$ and $(s(\ell_i\ldots\ell_j)^\omega,0) \models
\varphi$.  Now
\begin{eqnarray*}
 u & \sim_k & u_0u_1\ldots u_n(\ell_i\ldots\ell_j)^\omega \qquad
     \mbox{by Proposition~\ref{prop:bound}}\\
  & \sim_k & s(\ell_i\ldots\ell_j)^\omega \qquad
      \mbox{by Proposition~\ref{prop:comp}} \, .
\end{eqnarray*}
We conclude that $(u,0) \models \varphi$.

Similarly we can show that if  the highest infinitely often
occurring priority in a run of $\mathcal{A}_{\varphi}$ on $u$ is odd then
$(u,0) \not\models \varphi$.
\end{proof}

\begin{prop}
If $\varphi$ over set of monadic predicates $\mathcal{P}$ has
quantifier depth $k$, then $\mathcal{A}_\varphi$ has number of states
at most $2^{2^{O(|\mathcal{P}|k)}}$ and can be computed from $\varphi$
in time $|\varphi|^{O(1)} \cdot 2^{2^{O(|\mathcal{P}|k)}}$.
\end{prop}
\begin{proof}
The set of states $\mathrm{Rep}_k(\Sigma)$ has size at most
$2^{2^{O(|\mathcal{P}|k)}}$ and can be constructed in time at most
$2^{2^{O(|\mathcal{P}|k)}}$ by Corollary~\ref{corl:rep}.  We can establish the
existence of a transition between any pair of states of
$\mathcal{A}_\varphi$ in time at most $2^{O(|\mathcal{P}|k)}$ by
Proposition~\ref{prop:compute}.  Finally we can compute the priority
of a state $(s,\ell,i)$ by model checking $\varphi$ on a lasso of
length at most $2^{O(|\mathcal{P}|k)}$, which can be done in time
$|\varphi|^{O(1)} \cdot 2^{O(|\mathcal{P}|k)}$.
\end{proof}

\myparagraph{Extension to $\FO2$ with successor} We now extend to
successor using the same approach as in the proof of Theorem
\ref{BW:detLimit}.  By Lemma~\ref{thm:fo2utl}, given an $\FO2$ formula
$\varphi$ of quantifier depth $k$ there is an equivalent UTL formula
$\varphi'$ of at most exponential size and operator depth at most
$2k$. Moreover, $\varphi'$ can be transformed to a normal form such
that all next-time $\CIRC$ and last-time $\CIRCM$ operators are pushed
inside the other operators.  Again, we consider $\varphi'$ also as a
$\TL$-formula over an extended set of predicates $\mathcal{P}' =
\{\CIRC^iP_j, \CIRCM^iP_j \mid P_j \in P, i \le k\}$.  By a
straightforward transformation we get an equivalent $\FO2[<]$ formula
$\varphi ''$ over $P'$. Overall, this transformation creates
exponentially larger formulas, but the quantifier depth is only
doubled and the set of predicates is quadratic.  Applying Theorem
\ref{DRA} for $\varphi''$ over set of predicates $\mathcal{P}'$ gives:

\begin{thm}
\label{DRASUC}
Given an $\FO2$ formula $\varphi$ with quantifier depth $k$, there is
a deterministic parity automaton having $2^{2^{O(k^2|\mathcal{P}|)}}$ states and
$2^{O(k|\mathcal{P}|)}$ priorities that accepts the language $L(\varphi)$.
\end{thm}

\section{Models considered} \label{sec:models}

Next we collect together definitions of the various different types of
state machine that we consider in this paper. For non-deterministic
machines we will be interested in the existence of an accepting path
through the machine that satisfies a formula, while for probabilistic
models we want to know the probability of such paths.

\myparagraph{Kripke Structures, Hierarchical and Recursive State
  Machines} Our most basic model of non-deterministic computation is a
Kripke structure, which is just a graph with an additional set of
nodes (the initial states), and a labelling of nodes with a subset of
a collection of propositions.  The behavior represented by such a
structure is the set of paths through the graph, where paths can be
seen as $\omega$-words.

We will look also at more expressive and succinct structures for
representing behaviours.  A recursive state machine (RSM) $\M$ over a
set of propositions $\mathcal{P}$ is given by a tuple
$(M_1,\ldots,M_k)$ where each component state machine $M_i= (N_i \cup
B_i, Y_i, X_i, En_i, Ex_i, \delta_i)$ contains
\begin{iteMize}{$\bullet$}
\item a set $N_i$ of \emph{nodes} and a disjoint set $B_i$ of \emph{boxes};
\item an indexing function $Y_i : B_i \mapsto \{1, \ldots , k\}$ that
assigns to every box an index of one of the component machines,
$M_1, \ldots , M_k$; 
\item a labelling function $X_i: N_i \mapsto
2^{\mathcal{P}}$; 
\item a set of \emph{entry nodes} $En_i \subseteq N_i$ and a set
of \emph{exit nodes} $Ex_i \subseteq N_i$; 
\item a \emph{transition
relation} $\delta_i$, where transitions are of the form $(u, v)$ where
the source $u$ is either a node of $N_i$, or a pair $(b, x)$, where
$b$ is a box in $B_i$ and $x$ is an exit node in $Ex_j$ for $j =
Y_i(b)$.  We require that the destination $v$ be either a node in
$N_i$ or a pair $(b, e)$, where $b$ is a box in $B_i$ and $e$ is an
entry node in $En_j$ for $j = Y_i(b)$.
\end{iteMize}

\noindent Informally, an RSM represents behaviors that can transition
through a box into the entry node of the machine called by the box,
and can transition via an exit node back to the calling box, as with
function calls.  The semantics can be found in \cite{rsm}.  A
hierarchical state machine (HSM) is an RSM in which the dependency
relation between boxes is acyclic.  HSMs have the same expressiveness
as flat state machines, but can be exponentially more succinct.

\myparagraph{Markov Chains} The basic probabilistic model
corresponding to a Kripke structure is a (labelled) \emph{Markov
  chain}, specified as $\M=(\Sigma,X,V,E,M,\rho)$, consisting of an
\emph{alphabet} $\Sigma$, a set $X$ of \emph{states}; a
\emph{valuation} $V : X \rightarrow \Sigma$; a set $E \subseteq X
\times X$ of \emph{edges}; a \emph{transition probability} $M_{xy}$
for each pair of states $(x,y)\in E$ such that for each state $x$,
$\sum_y M_{xy} = 1$; an \emph{initial probability distribution} $\rho$
on the set of states $X$.

A Markov chain defines a probability distribution on
trajectories---paths through the chain.  Given a language $L \subseteq
\Sigma^\omega$, we denote by $P_\M(L)$ the probability of the set of
trajectories of $\M$ whose image under $V$ lies in $L$.  We consider
the complexity of the following model checking problem: Given a Markov
chain $\M$ and an LTL- or FO$^2$-formula $\varphi$, calculate
$P_\M(L(\varphi))$.  There is a decision version of this problem that
asks whether this probability exceeds a given rational threshold.

\myparagraph{Recursive Markov Chains} Recursive Markov chains (RMCs)
are the analog of RSMs in the probabilistic context. They are defined
as RSMs, except that the transition relation consists of triples $(u,
p_{u,v}, v)$ where $u$ and $v$ are as with RSMs, and the $p_{u,v}$ are
non-negative reals with $\Sigma_v p_{u,v}=1$ or $0$ for every $u$. As
with Markov chains, these define a probability distribution on
trajectories, but now trajectories are paths which must obey the
box-entry/box-exit discipline of an RSM.  The semantics of an RMC can
be found in \cite{rmc}.  A hierarchical Markov chain (HMC) is the
probabilistic analog of an HSM, that is, an RMC in which the calling
graph is acyclic. An HMC can be converted to an ordinary Markov chain
via unfolding, possibly incurring an exponential blow-up.
An example of an RMC is shown in Figure \ref{fig:rmc}.

\begin{figure}[h]
\begin{center}
\scalebox{1}{

\includegraphics{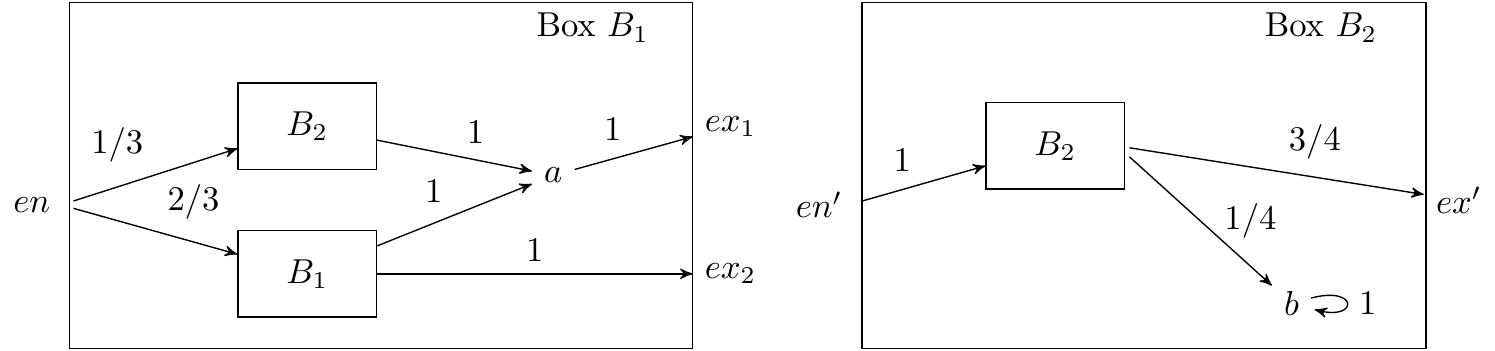}

}
\end{center}
\caption{A sample Recursive Markov Chain}
\label{fig:rmc}
\end{figure}

\myparagraph{Markov Decision Processes} We will also deal with
verification problems related to control of a probabilistic process by
a scheduler.  A \emph{Markov decision process (MDP)}
$\M=(\Sigma,X,N,R,V,E,M,\rho)$ consists of an \emph{alphabet}
$\Sigma$, a set $X$ of \emph{states}, which is partitioned into a set
$N$ of \emph{non-deterministic states} and a set $R$ of
\emph{randomising states}; a \emph{valuation} $V : X \rightarrow
\Sigma$, a set $E \subseteq X \times X$ of \emph{edges}, a
\emph{transition probability} $M_{xy}$ for each pair of states
$(x,y)\in E$, $x \in R$ such that $\sum_y M_{xy} = 1$; an
\emph{initial probability distribution} $\rho$.  This model is
considered in~\cite{CY95} under the name \emph{Concurrent Markov
  chain}.

We can view non-deterministic states as being controlled by the
scheduler, which given a trajectory leading to a non-deterministic
state $s$ chooses a transition out of $s$.  There are two basic
qualitative model checking problems: the \emph{universal problem}
($\forall$) asks that a given formula be satisfied with probability
one for all schedulers; the \emph{existential problem} ($\exists$)
asks that the formula be satisfied with probability one for some
scheduler.  The latter corresponds to the problem of designing a
system that behaves correctly in a probabilistic environment.  In the
\emph{quantitative model checking problem}, we ask for the maximal
probability for the formula to be satisfied on a given MDP when the
scheduler chooses optimal moves in the non-deterministic states.

\myparagraph{Two-player Games}
A \emph{two-player game} $G=(\Sigma,X,X_1,X_2,V,E,x_0)$ consists of 
an \emph{alphabet} $\Sigma$;
a set $X$
of \emph{states}, which is partitioned into a set $X_1$ of states
controlled by \emph{Player~I} and a set $X_2$ controlled
by \emph{Player II}; a set of $E \subseteq X \times X$
of \emph{transitions}; a \emph{valuation} $V:X \rightarrow \Sigma$; an
\emph{initial state} $x_0$.

The game starts in the initial state and then the player who controls
the current state, taking into account the whole history of the game,
chooses one of the possible transitions.  The verification problem of
interest is whether Player I has a strategy such that for all infinite
plays the induced infinite word $u \in \Sigma^\omega$ satisfies a
given formula $\varphi$.

\myparagraph{Stochastic Two-player Games} A \emph{Stochastic
  two-player game} ($2\frac{1}{2}$-player game)
$G=(X,X_1,X_2,R,V,E,M,p_0)$ consists of a set $X$ of \emph{states},
which is partitioned into a set $X_1$ of states controlled by
\emph{the first player}, a set $X_2$ controlled by \emph{the second
  player} and a set $R$ of \emph{randomising states}; a
\emph{valuation} $V:X \rightarrow \Sigma$; a set of $E \subseteq X
\times X$ of \emph{transitions}, a \emph{transition probability}
$M_{xy}$ for each pair of states $(x,y)\in E$, $x \in R$ such that
$\sum_y M_{xy} = 1$; an \emph{initial probability distribution}
$\rho$. See Figure \ref{fig:stochasticgame} for an example.

The \emph{universal} ($\forall$) qualitative model checking problem asks if the first player can 
enforce that the infinite word $u$, induced by the path 
through the game, satisfies $\varphi$ with probability one.

\begin{figure}[h]
\begin{center}
\scalebox{1}{

\includegraphics{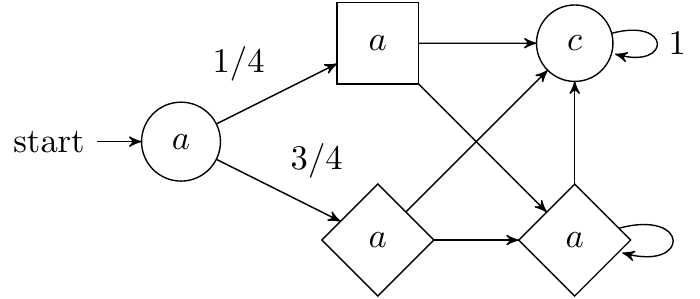}

}
\end{center} 
\caption{A sample Stochastic Two-player Game. Diamonds are states of the first player, squares are states of the second player and circles represent randomising states.}
\label{fig:stochasticgame}
\end{figure}

\section{Verifying non-deterministic systems}
\label{sec:nondet}
Model checking for traditional Kripke Structures is fairly
well-understood. All of our logics subsume propositional logic, and
the model checking problems we deal with generalise propositional
satisfiability---hence they are all NP-hard.  LTL and UTL are both
PSPACE-complete \cite{SC85}, while ($\TL$) is NP-complete.

Translation I shows how to convert an $\fotwo$ formula to a union of
exponential sized automata.  A NEXPTIME algorithm can guess such an
automaton, take its product with a given Kripke Structure, and then
determine non-emptiness of the resulting product. Coupled with the
hardness argument in \cite{fo2_utl}, this gives an alternative proof
of the result of Etessami, Vardi, and Wilke:
\begin{thm} \label{thm:kripke} 
\cite{fo2_utl}
$\fotwo$ model-checking is complete for
NEXPTIME.
\end{thm}

Below we extend these results to give a comparison of the complexity
of model checking for recursive state machines and two-player games,
applying all of the translations in the previous section.

\subsection{Recursive State Machines}
Using Translation II, we show that $\FOtwoLT$ model checking  can be done as efficiently as for $\TL$ on
non-deterministic systems, and in particular for RSMs.

\begin{prop} \label{npbound}
Model checking $\FOtwoLT$ properties on Kripke structures,
hierarchical and recursive state machines is in NP.
\end{prop}
\begin{proof}
We give the upper bound for RSMs only, since the other classes are special cases.
We describe an NP algorithm that checks satisfiability of an $\FOtwoLT$
sentence $\varphi$ on the language of RSM $\M$.   Model checking
the structure involves only combinations of propositions occurring in the structure,
and hence by expanding out these combinations explicitly, we can assume that
the unary alphabet restriction holds.  Thus we can
apply Translation II, from $\FOtwoLT$ to B\"uchi
Automata, Theorem
\ref{FO2automaton}. It suffices to
check
that one of the
 automata $A_i$ produced by the translation
accepts a word produced by $\M$. We can thus guess such
an $A_i$ 
and can then check intersection
of $A_i$ with $\M$ in polynomial time, 
by forming the product and checking that we can reach an accepting bottom strongly 
connected component. This reachability analysis can be done efficiently using
the ``summary edge construction''---see, e.g.,  \cite{rsm}.
\end{proof}

In the same way, we can obtain the result for model checking full $\FOtwo$ on RSMs, 
but now using the $\FO2$ to automata translation in Translation 1, Theorem \ref{thm:FO2_aut}.
Again we guess an automata $A_i$, which is now of exponential size. Thus we have:

\begin{prop} \label{rsmfo2}
$\FOtwo$ model checking of RSMs can be done in NEXPTIME.
\end{prop}
This result matches the known result for ordinary Kripke structures.

\subsection{Two-player games with FO$^2$ winning condition}
Two-player games are known to be in 2EXPTIME for LTL \cite{PnueliR89}.
We now show that the same is true for $\FO2$, making use of 
Translation III in the previous section,
which translates to deterministic parity automata.  We also utilise the fact that a
parity game with $n$ vertices, $m$ edges and $d$ priorities can be
solved in time $O(dmn^d)$~\cite{J00}.

From these two results we easily conclude the 2EXPTIME upper bound:
\begin{prop}
\label{2playerFO2game}
Two-player games with $\FO2$ winning conditions are solvable in 2EXPTIME.
\end{prop}
\begin{proof}
Using Theorem \ref{DRASUC}, we construct in 2EXPTIME a deterministic
parity automaton for the $\FO2$ formula $\varphi$ with doubly
exponentially many states and at most exponentially many priorities.
By taking the product of this automaton with the graph of the game, we
get a parity game with doubly exponentially many states but only
exponentially many priorities. (In fact if we define the automaton
over an alphabet $\Sigma \subseteq 2^{\mathcal{P}}$ containing only sets of
propositions that occur as labels of states in the game, then
polynomially many priorities suffice.)  We can then determine the
winner in double exponential time, again applying the $O(dmn^d)$ bound
for solving games of ~\cite{J00} mentioned above.
\end{proof}

Combining this with the result by Alur, La Torre, and Madhusudan, who
showed that two-player games are
2EXPTIME-hard \cite{gamesBoxesDiamonds} already for the simplest $\TL$, along with the
fact that we can convert UTL formula to $\FO2$ formula in polynomial
time, we get 2EXPTIME-completeness:
\begin{cor}
Deciding two-player games with $\FOtwo$ winning conditions is  complete for 2EXPTIME.
\end{cor}

The table below summarises both the known results and the results from this paper (in bold) concerning non-deterministic systems. All bounds are tight.
{
\renewcommand{\tabcolsep}{1mm}
\begin{small}
\begin{figure}[h!]
\begin{center}
\begin{tabular}{|l|c|c|c|c|c|}
\hline
       &$\TL$&$\UTL$&$\FOtwoLT$&$\FOtwo$&$\LTL$\\
\hline
Kripke Structure        & NP  & PSPACE & \CC NP  & NEXP          &PSPACE\\
HSM   & NP & PSPACE &  \CC NP & \CC NEXP & PSPACE\\
RSM          &  NP & EXP  &  \CC NP& \CC NEXP & EXP \\
Two pl. games          &  2EXP & 2EXP  &  \CC 2EXP& \CC 2EXP & 2EXP \\
\hline
\end{tabular}
\end{center}
\end{figure}
\end{small}
}

The PSPACE bound for model checking LTL on HSMs follows by expanding
the HSMs to `flat' Kripke structures and recalling that model checking
LTL on Kripke structures can be done in space polynomial in the
logarithm of the model size. Additionally, the complexity of model
checking UTL and LTL on RSMs is
EXPTIME-complete \cite{PDA-Reachability}, and model checking $\TL$ on
RSMs is NP-complete \cite{LTLonRSM}.

\section{Verifying probabilistic systems}
\label{sec:prob}
We now turn to probabilistic systems. Here we will make use of two key
properties of the automata produced by the first two
translations---unambiguity and determinism in the limit.  We will need
two lemmas, which show that the complexity bounds for model checking
unambiguous B\"uchi automata on various probabilistic systems are the
same as the bounds for deterministic B\"uchi automata on these
systems. First, following~\cite{ltltosepaut}, we note the following
property of unambiguous automata:

\begin{lem}
\label{BW:MCHlemma}
Given a Markov chain $\M=(\Sigma,X,V,E,M,\rho)$ and a generalised
B\"{u}chi automaton $A=(\Sigma,S,S_0,\Delta,\lambda,\mathcal{F})$ that is
unambiguous, $P_\M(L(A))$ can be computed
in time polynomial in $\M$ and $A$.
\end{lem}
\proof
We define a directed graph $\M \otimes A$ representing the
synchronised product of $\M$ and $A$.  The vertices of $\M \otimes A$
are pairs $(x,s) \in X \times S$ with matching propositional labels,
i.e., such that $V(x)=\lambda(s)$; the set of directed edges is
$\{((x,s),(y,t)) : (x,y) \in E \mbox{ and } (s,t) \in \Delta\}$.  We
say that a strongly connected component (SCC) of $\M \otimes A$ is
\emph{accepting} if (i) for each set of accepting states $F \in
\mathcal{F}$ it contains a pair $(x,s)$ with $s \in F$ and (ii) for
each pair $(x,s)$ and each transition $(x,y) \in E$, there exists
$(s,t) \in \Delta$ such that $(y,t)$ is in the same SCC as $(x,s)$.
This guarantees that we can stay in the SCC and visit each of its
states infinitely often.

Let $L(A,s)$ denote the set of words accepted by $A$ starting in state $s$.
For each vertex $(x,s)$ of $\M \otimes A$ we have a variable
$\xi_{x,s}$ representing the probability $P_{\M,x}(L(A,s))$ of all runs
of $\M$ starting in state $x$ that are in $L(A,s)$.  These
probabilities can be computed as the unique solution of the following
linear system of equations:
\begin{eqnarray*}
\xi_{x,s}  & = & 1 \qquad \mbox{$(x,s)$ in an accepting SCC}\\
\xi_{x,s}  & = & 0 \qquad \mbox{$(x,s)$ in a non-accepting SCC}\\
\xi_{x,s}  & = & \sum_{(s,t)\in\Delta}\;
                \sum_{y:V(y)=\lambda(t)}
               M_{xy} \cdot \xi_{y,t} \qquad \mbox{ otherwise.}
               \end{eqnarray*}
The correctness of the third equation follows from the following calculation:
\begin{eqnarray*}
P_{\M,x}(L(A,s))
          & = & P_{\M,x}(\;\bigcup_{(s,t)\in\Delta} 
                \lambda(s)\cdot L(A,t)\;)\\
          & = & \sum_{(s,t)\in \Delta} P_{\M,x}(\lambda(s)\cdot L(A,t))
\;\;\mbox{ (since $A$ is unambiguous)} \\
          & = & \sum_{(s,t)\in\Delta}\,\sum_{y:V(y)=\lambda(t)}
                M_{xy} \cdot P_{\M,y}(L(A,y)) \, .
\rlap{\hbox to 115 pt{\hfill\qEd}}
\end{eqnarray*}

\noindent For an RMC $\M$, we can compute reachability probabilities
$q_{(u,ex)}$ of exiting a component $M_i$ starting at state $u \in
V_i$ going to exit $ex \in Ex_i$.  Etessami and Yannakakis \cite{rmc}
show that these probabilities are the unique solution of a system of
non-linear equations which can be found in polynomial space using a
decision procedure for the existential theory of the reals. Following
\cite{rmc} for every vertex $u \in V_i$ we let $ne(u)=1-\sum_{ex \in
  Ex_i}q_{(u,ex)}$ be the probability that a trajectory beginning from
node $u$ never exits the component $M_i$ of $u$.  Etessami and
Yannakakis~\cite{rmc_mc} also show that one can check properties
specified by deterministic B\"uchi automata in PSPACE, while for
non-deterministic B\"uchi automata they give a bound of EXPSPACE.
Thus the prior results would give a bound of EXPSPACE for UTL and
2EXPSPACE for $\FO2$. We will improve upon both these bounds.  We
observe that the technique of~\cite{rmc_mc} can be used to check
properties specified by non-deterministic B\"uchi automata that are
unambiguous in the same complexity as deterministic ones. This will
then allow us to apply our logic-to-automata translations.

\begin{prop} \label{thm:rmcsepba} 
Given an unambiguous B\"uchi automaton $A$ and a RMC $\M$, we can compute
the probability that $A$ accepts a trajectory of $\M$ in PSPACE.
\end{prop}

\begin{proof}
Let $A$ be an unambiguous B\"uchi automaton with set of states $Q$,
transition function $\Delta$ and labelling function $\lambda$.  Let $\M$
be an RMC with valuation $V$.  We define a product RMC $\M \otimes A$
with component and call structure coming from $\M$ whose states are
pairs $(x,s)$, with $x$ a state of $\M$ and $s$ a state of $A$ such
that $V(x)=\lambda(s)$ (i.e., $x$ and $s$ have the same label).  Such a
pair $(x,s)$ is accepting if $s$ is an accepting state of $A$.  A run
through the product chain is accepting if at least one of the
accepting states is visited infinitely often.  Note that a path
through $\M$ may expand to several runs in $\M \otimes A$ since $A$ is
non-deterministic.

For each $i$, for each vertex $x\in V_i$, exit $ex \in Ex_i$ and
states $s,t \in Q$ we define $p(x,s \rightarrow ex,t)$ to be the
probability that a trajectory in RMC $\M$ that begins from a
configuration with state $x$ and some non-empty context (i.e. not at
top-level) expands to an accepting run in $\M \otimes A$ from $(x,s)$ to
$(ex,t)$.  

Just as in the case of deterministic automata, we can compute $p(x,s
\rightarrow ex,t)$ as the solution of the following system of
non-linear equations:

If $x \in V_i$ is not entrance of the box we have:
$$p(x,s \rightarrow ex,t) = \sum_{x':(x,M_{xx'},x') \in \delta_i}
M_{xx'} \sum_{s':(s,s') \in \Delta \wedge \lambda(s')=V(x')} p(x',s'
\rightarrow ex,t)
$$

If $x \in V_i$ is entrance of the box $b \in B_i$ then we include the
equations:
$$p(x,s \rightarrow ex,t) = \sum_{j, s' \in Q} 
p((b,en),s \rightarrow (b,ex_j),s') p((b,ex_j),s' \rightarrow ex,t)$$
where $p((b,en),s \rightarrow (b,ex_j),s') = p(en_{Y_i(b)},s \rightarrow ex_j, s')$ and $ex_j \in Ex_{Y_i(b)}$.

The justification for these equations is as follows.  Since $A$ is
unambiguous, each trajectory of $\M$ expands to at most one accepting
run of $\M \otimes A$.  Thus in summing over automaton states $s'$ in
the two equations above we are summing probabilities over disjoint
events which correctly gives us the probability of the union of these
events.
 
We now explain how these probabilities can be used to compute the
probability of acceptance.  We assume without loss of generality that
the transition function of $A$ is total.

We construct a finite-state \emph{summary chain} for the product
$\M \otimes A$ exactly as in the case of deterministic
automata~\cite{rmc_mc}.  For each component $M_i$ of $\M$, vertex $x$
of $M_i$, exit $ex \in Ex_i$ and for each pair of states $s,t$ of $A$
the probability to transition from $(x,s)$ to $(ex,t)$ in the summary
chain is calculated from $p(x,s \rightarrow ex,t)$ after adjusting for 
probability $ne(x)$ that $\M$ never exits $M_i$ starting at vertex $x$.
Note that since automaton $A$ is non-blocking, the probability of
never exiting the current component of $\M \otimes A$ starting at
$(x,s)$ is the same as $ne(x)$ (the probability of never exiting the
current component from vertex $x$ in the RMC $\M$ alone).  

To summarise, we first compute reachability probabilities $q_{(u,ex)}$
and probabilities $ne(u)$ for the RMC $\M$. Then we consider the
product $\M \otimes A$ and solve a system of non-linear equations to
compute the probabilities of summary transitions $p(x,s \rightarrow
ex,t)$.  From these data we build the summary chain, identify
accepting SCCs and compute the resulting probabilities in the same way
as in \cite{rmc_mc}. All these steps can be expressed as a formula and
its truth value can be decided using existential theory of the reals
in PSPACE.
\end{proof}

\subsection{Markov chains} \label{subsec:markov_chains}

We are now ready to prove a new bound for the model checking problem
on our most basic probabilistic system,  Markov chains.
  Courcoubetis and Yannakakis \cite{CY95} showed that
one can determine if an LTL formula holds with non-zero probability in
a Markov chain in PSPACE\@.  This gives a PSPACE upper bound for $\TL$
and an EXPSPACE upper bound for $\FOtwo$.  We will show how to get
better bounds, even in the
quantitative case, using the logic-to-automata translations.

\begin{prop}
\label{thm:fo2_mc}
Model checking $\TL$ or $\FOtwoLT$ on Markov chains is in \#P.
\end{prop}

\begin{proof}
Let $\varphi$ be a $\TL$ or $\FOtwoLT$ formula and $\M$ a Markov chain.  Using
Theorem~\ref{BW:detLimit} in case of $\TL$ and
Theorem~\ref{FO2automaton} in case of $\FOtwoLT$, we have that for formula $\varphi$ there
is a family $\{A_i\}$ comprising at most
$2^{\mathit{poly}(|\varphi|,|\Sigma|)}$ unambiguous generalised
B\"{u}chi automata, whose languages partition $\{ w \in \Sigma^\omega
: w \models \varphi\}$.  Moreover, each $A_i$ has at most
$|\varphi||\Sigma|$ states and can be generated in polynomial time
from $\varphi$ and index $i$.  By Lemma 
\ref{BW:MCHlemma}
we can
further compute the probability $p_i$ of $\M$ satisfying $A_i$ in
polynomial time in the sizes of $\M$ and $A_i$.  Since each $p_i$ is
computable in polynomial time we can determine $\sum_i p_i$ in $\#P$.
\end{proof}

\begin{prop} \label{thm:mcfo2mc}
The threshold problem for model checking $\FOtwo$  on
Markov chains is in PEXP.
\end{prop}
\begin{proof}
The result follows by the same argument as in
Proposition \ref{thm:fo2_mc}, as we are essentially in the same situation,
but now by Theorem \ref{thm:FO2_aut} we have a collection of doubly-exponentially many automata, each of
exponential size. 
\end{proof}

\subsection{Hierarchical and Recursive Markov chains}

Similarly, we get the following results for recursive Markov chains (and in particular
for  hierarchical Markov chains):

\begin{prop} \label{thm:rmctl} The probability of a $\TL$ or $\FOtwoLT$ formula holding on a recursive Markov chain can be computed in PSPACE.
\end{prop}
\begin{proof}
By Theorem~\ref{BW:detLimit} in case of $\TL$ and
by Theorem~\ref{FO2automaton} in case of $\FOtwoLT$, we can convert a formula $\varphi$
into an equivalent disjoint union of exponentially many unambiguous
automata of polynomial size (in $|\varphi|$ and $|\Sigma|$) and the RMC.  Using
polynomial space we can generate each automaton, calculate
the probability that the RMC generates an accepting trajectory by Proposition~\ref{thm:rmcsepba} , and
sum these probabilities for each automaton.
\end{proof}

\begin{cor} \label{cor:rmctl} The probability of a $\TL$ or $\FOtwoLT$ formula holding on a 
hierarchical Markov chain can be computed in PSPACE.
\end{cor}

\begin{prop} \label{thm:rmcfo2} The probability of an  $\FO2$ formula holding on an RMC can be computed in EXPSPACE.
\end{prop}
\begin{proof}
The result follows by the same argument as in Proposition \ref{thm:rmctl}, but now by Theorem \ref{thm:FO2_aut}
we have family of doubly exponentially many automata each of exponential size, with a non-deterministic
exponential time algorithm for building each automaton.
  Therefore applying 
Proposition \ref{thm:rmcsepba} we immediately obtain upper bounds for
$\FO2$.
\end{proof}

For an ordinary Markov chain, calculating the probability of an LTL
formula can be done in PSPACE \cite{mihalisprivate}, while we have
seen previously that we can calculate the probability of an $\FO2$
formula in PEXP. One can achieve the same bounds for LTL and $\FO2$ on
hierarchical Markov chains.  In each case we expand the HMC into
an ordinary Markov chain and then use the model checking algorithm for
a Markov chain.  This does not impact the complexity, since the space
complexity is only polylog in the size of the machine for LTL and the
time complexity is only polynomial in the machine size for $\FOtwo$. We
thus get:

\begin{prop} \label{thm:hmcutlfo2} 
The probability of a $\FO2$ formula holding on a HMC can be computed
in PEXP, while for an LTL formula it can be computed in PSPACE.
\end{prop}

\subsection{Markov decision processes} 
\label{subsec:markov_decision_processes}
Courcoubetis and Yannakakis~\cite{CY95} have shown that the maximal
probability with which a scheduler can achieve an UTL objective on an
MDP can be computed in 2EXPTIME.  It follows from results
of~\cite{gamesBoxesDiamonds} that even the qualitative problem of
determining whether every scheduler achieves probability~$1$ is
2EXPTIME-hard.  Combining the 2EXPTIME upper bound with the
exponential translation from $\FO2$ to UTL~\cite{fo2_utl} yields a
3EXPTIME bound for $\FO2$.  Below we see that using our
$\FO2$-to-automaton construction we are able to improve this bound to
2EXPTIME.

We begin with universal formulation of qualitative model checking
MDPs. To deal with MDP's, we will make use of  determinism in the limit.
\begin{prop}
\label{thm:qual_mdp_fo2}
Determining whether for all schedulers a $\FOtwoLT$-formula $\varphi$
holds almost surely on a Markov decision process $\M$ is
co-NP-complete.
\end{prop}
\begin{proof}
The corresponding complement problem asks whether there exists a
scheduler $\sigma$ such that the probability of $\neg \varphi$ is
greater than 0. For this problem, there is an NP algorithm, as we now
explain.  In Courcoubetis and Yannakakis \cite{CY95}, there is a
polynomial time algorithm for qualitative model checking deterministic
B\"{u}chi automata on MDPs. As noted there, the algorithm applies to
automata that are deterministic in the limit as well.  Therefore we
can just guess a particular automaton $A_i$ from the family of
automata corresponding to $\neg\varphi$, as described in
Theorem~\ref{FO2automaton}.  The theorem guarantees that this
automaton will be deterministic in the limit.

It is easy to see that the co-NP is tight, even for $\TL$, since
qualitative model checking for MDPs generalises validity for both
$\TL$ formulas, which is co-NP hard.
\end{proof}

\begin{prop} \label{thm:qual_mdp_utl_fo2}
Determining whether for all schedulers a UTL-formula $\varphi$ holds
almost surely on a Markov decision process $\M$ is in EXPTIME.  For
$\FO2$ the problem is in co-NEXPTIME.
\end{prop}
\begin{proof}
The result for $\FO2$ follows along the lines of the proof of
Proposition \ref{thm:qual_mdp_fo2}, but now we guess an automaton $A_i$ of
exponential size (using Theorem \ref{thm:FO2_aut}).

Similarly, for UTL we can use Theorem \ref{BW:detLimit}. We still have
exponential sized automata $A_i$, but only exponentially many of them,
so we can iterate over all of them, which gives us a single
exponential algorithm.
\end{proof}

Note that here the $\FO2$ problem is \emph{easier} than the
corresponding LTL problem, which is known to be 2EXPTIME-complete.

For the existential case of the qualitative model-checking problem, an
upper bound of 2EXPTIME for all of our languages will follow from the
quantitative case below. On the other hand the arguments
from~\cite{gamesBoxesDiamonds} can be adapted to get a 2EXPTIME lower
bound (see Proposition \ref{lower_tl_mdp}) even for qualitative model-checking $\TL$ in the existential
case.  Hence we have:
\begin{prop} \label{thm:qualexistsmdp} Determining if there is a scheduler
that enforces a formula with probability one is 2EXPTIME-complete for 
each of $\TL$, $\rm{UTL}$, $\rm{LTL}$, $\FOtwoLT$ and  $\FOtwo$.
\end{prop}

We now turn to the quantitative case.  We apply the translation
from $\FO2$ to deterministic parity automata from
Subsection~\ref{subsec:trans}, along with the result that the value of a
Markov decision process with parity winning objective can be computed
in polynomial time~\cite{surveyStochGames}.  Using Theorem~\ref{DRASUC}
we immediately get bounds for $\FO2$ that match the known bounds for
LTL:

\begin{prop} \label{thm:quant_mc_fo2}
We can compute the maximum probability of an $\FO2$ formula $\varphi$
over all schedulers on a Markov decision processes $\M$ in 2EXPTIME.
\end{prop}

\subsection{Stochastic two-player games with $\FOtwo$ winning condition}

We can reduce the qualitative case of stochastic two-player games to the case of ordinary
two-player games using the following result of Chatterjee, Jurdzinski and Henzinger:

\begin{prop}[\cite{ChatterjeeSimpleStoch}]
\label{StochasticParity}
Every (universal) qualitative simple stochastic parity game with $n$
vertices, $m$ edges and $d$ priorities can be translated to a simple
parity game with the same set of priorities, with O($d n$) vertices
and O($d (m+n)$) edges, and hence it can be solved in time O($d (m +
n) (nd) ^ {d/2}$).
\end{prop}

Now combining the reduction with our results for two-player games, we
ascertain the complexity of stochastic two-player games:

\begin{cor}
The universal qualitative model checking problem for Stochastic
two-player games ($2\frac{1}{2}$-player game) with $\FOtwo$ winning
condition is 2EXPTIME-complete.
\end{cor}
\begin{proof}
Hardness follows from 2EXPTIME-hardness for two-player games with
$\FOtwo$ winning conditions. Membership is a consequence of the above
reduction and our bounds for two-player games (see Proposition
\ref{2playerFO2game} and Proposition \ref{StochasticParity}).
\end{proof}

\subsection{Lower bounds}
We can get corresponding tight lower bounds for most of the probabilistic model
checking problems.

\begin{prop}
The quantitative model checking problem for a $\TL$ formula $\psi$ on
a Markov chain $\M$ is \#P-hard.\label{UTLonMChardness}
\end{prop}

\begin{proof}
The proof is by reduction from \#SAT. Let $\varphi$ be a propositional
formula over literals $a_1, a_2, \ldots a_n$. We construct a Markov
chain $\M$ such that each trajectory generated by $\M$ corresponds to an
assignment of truth values to literals $a_1, \ldots a_n$, with each of
the $2^n$ possible truth assignments arising with equal probability.
We also construct a $\TL$ formula $\psi$ such that only trajectories
of $\M$ that encode satisfying valuations contribute to the probability
$P_\M(L(\psi))$. Therefore the number of satisfying assignments of the
original propositional formula $\varphi$ is $2^nP_\M(L(\psi))$.

See Figure~\ref{mcUTLX} for a depiction of the Markov chain $\M$ in case
$n=3$. All probabilities equal $1/2$, except those on transitions
leading to the final vertex~$f$. A path going through vertex $a_i$
corresponds to assigning true to the literal $a_i$ and a path through
$a_i'$ to an assignment of false.
We construct the $\TL$ formula $\psi$ corresponding to the
propositional formula $\varphi$ by replacing each positive literal
$a_i$ in $\varphi$ with $\DIAMOND a_i$ and each negative literal $\neg
a_i$ in $\varphi$ with $\DIAMOND a_i'$.

\begin{figure}[htb]
\begin{center}

\includegraphics{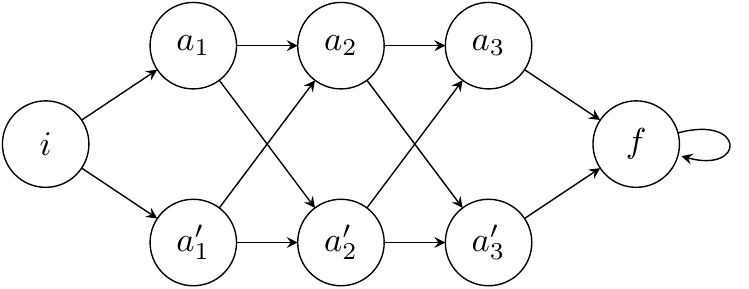}

\caption{Markov chain $\M$ for $n=3$}
\label{mcUTLX}
\end{center}
\end{figure}

\noindent Recalling the upper bound from Proposition \ref{thm:fo2_mc}, we
conclude that the quantitative model checking problem for $\TL$ on
Markov chains is \#P-complete.
\end{proof}

\begin{prop}
\label{MCFO2hard}
The quantitative model checking problem for $\FOtwo$ on Markov chains
is PEXP-hard.
\end{prop}

\begin{proof}
PEXP-hardness is by reduction from the problem of whether a strict
majority of computation paths of a given non-deterministic EXPTIME
Turing machine $T$ on a given input $I$ are accepting.  The Markov
chain generates a uniform distribution over strings of the appropriate
length, and the formula checks whether a given string encodes an
accepting computation of $\M$. The ability of $\FOtwo$ to check
validity of such a string has already been exploited in the
NEXPTIME-hardness proof for $\FO2$ satisfiability in \cite{fo2_utl}.
The details of this approach can be found in the proof of Proposition
\ref{fo2_let_mc_hard}.

Combining with the upper bound from Proposition \ref{thm:mcfo2mc}, the quantitative model checking problem for $\FOtwo$ on Markov chains
is PEXP-complete.
\end{proof}

Turning to lower bounds for MDPs, note that co-NEXPTIME-hardness for $\FO2$ is
inherited from the lower bound for Markov chains. On the other hand, we can show
that the EXPTIME bound for UTL is tight:

\begin{prop} \label{thm:utlmdpexp}
Determining whether for all schedulers a UTL-formula $\varphi$ holds
almost surely on a Markov decision process $\M$ is EXPTIME-hard.
\end{prop}

\begin{proof}
The argument is based on the idea of Courcoubetis and Yannakakis for
lower bounds in the LTL case.  We reduce the acceptance problem for an
alternating PSPACE Turing machine to the problem of whether there is a
scheduler that enforces that a UTL formula $\varphi$ holds with
positive probability.  Thus we reduce to the complement of the problem
of interest.

Consider an alternating PSPACE Turing machine $T$ with input $I$.
Without loss of generality we assume that each configuration of $T$
has exactly two successors and that $T$ uses space at most $n$ on an
input $I$ of length $n$.  Then we can encode a branch of the
computation tree of $T$ as a finite string in which each configuration
is represented by a consecutive block of $n+1$ letters: one bit to
represent the choice to branch left or right, and $n$ letters to
represent the configuration.  Let $L_{T(I)}$ be the language of
infinite strings, each of which is an infinite concatenation of finite
strings that encode accepting computations.  It is standard that one
can write a UTL formula $\varphi$ that captures $L_{T(I)}$.

Next we describe the MDP $\M$.  Intuitively the goal of the scheduler
is to choose a path through $\M$ so as to generate a word in $L_{T(I)}$.
A high-level depiction of $\M$ is given in Figure~\ref{mc5}. The
boxes \emph{init-conf} and
\emph{next-conf} contain gadgets that are used by the scheduler
to generate the initial configuration and all successive
configurations of $T$ as strings of length $n$.  The number of such
strings is exponential in $n$, but clearly the gadgets can be
constructed using only linearly many states.  After producing an
existential configuration of the Turing Machine, the scheduler sends
control to the state \emph{sch}, where it decides whether $T$
should branch left or right.  After  generating a universal
configuration, an~honest scheduler sends control to \emph{pro}, the
only randomising state in $\M$, where the branching direction $T$ is
selected uniformly at random.  When the scheduler has successfully
generated an accepting computation it visits \emph{acc}, which is the~only 
accepting state of $\M$, and the simulation starts over again from
the beginning.  Only those computations that visit \emph{acc}
infinitely often and in which the scheduler behaves honestly satisfy
$\varphi$.

We claim that there exists a scheduler such that $P_\M(L(\varphi))>0$
if and only if $T$ accepts its input.

If the Turing Machine $T$ accepts its input, then the scheduler can
simply follow the strategy from the alternating
computation of $T$.  Regardless of the choice made by the
probabilistic opponent, the scheduler can always go to an accepting
vertex with probability $1$. Therefore even if we repeat the whole
simulation, for this scheduler $P_M(L(\varphi))=1$, which is greater
than $0$ as required.

The infinite repetition is important in the second case, when the Turing
Machine $T$ rejects its input. If the process ran only once, it could
happen that in the probabilistic choice, only one option would lead to
a rejecting state, but it would not be chosen if the probabilistic
opponent of the scheduler were unlucky.  Therefore we repeat this
process infinitely many times and thus guarantee that with probability
$1$ we will reach the rejecting vertex and then stay there forever,
i.e. $P_\M(L(\varphi))$ will be $0$ as required.

Combining with the upper bound from Proposition
\ref{thm:qual_mdp_utl_fo2}, determining whether for all schedulers a
UTL-formula holds with probability one on a Markov decision process is
EXPTIME-complete.
\begin{figure}[h!]
\begin{center}

\includegraphics{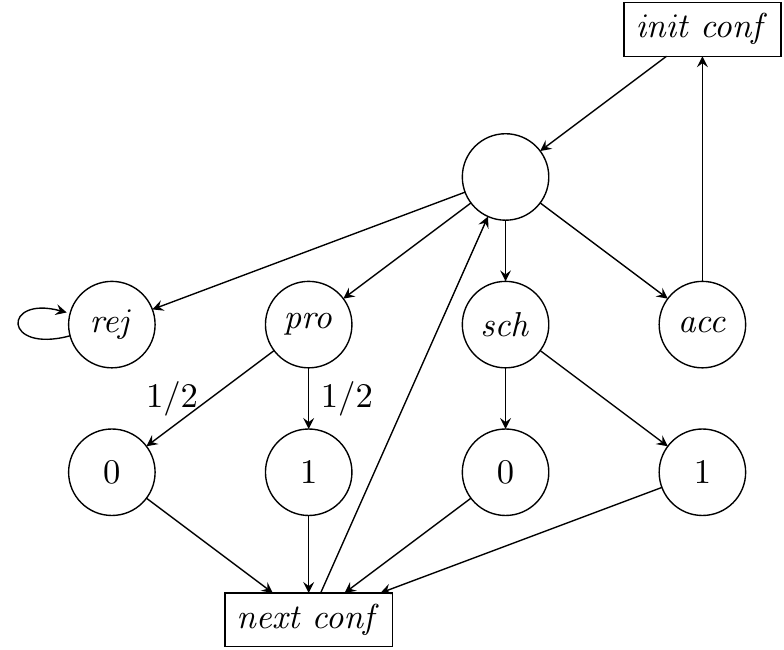}

\caption{Sketch of the Markov decision process $\M$}
\label{mc5}
\end{center}
\end{figure}
\end{proof}

The above was a lower bound for checking whether all schedulers enforce the
property with probability 1.
We now show a tight lower-bound for the existence of a probability one scheduler:

\begin{prop}
\label{lower_tl_mdp}
Given a Markov decision process and a $\TL$ formula, determining whether 
the formula holds with probability one for some scheduler is 2EXPTIME-hard.
\end{prop}
\begin{proof}
The proof is an adaptation of the 2EXPTIME-hardness proof of Alur
\emph{et.\ al.} for model checking $\TL$ formulas on two-player games
in \cite{gamesBoxesDiamonds}. The proof there is based on a reduction
from the membership problem for an alternating exponential-space
Turing machine, where a game graph and a $\TL$ formula are constructed
such that the Turing machine accepts the given input if and only if the
existential player has a winning strategy in the game.

We can adapt the proof by assigning the existential vertices of the
game graph to a scheduler and assigning the universal vertices from
the game graph to the probabilistic player (by setting the uniform
outgoing probabilities from these vertices).  When the Turing machine
accepts its input we are guaranteed that there is a corresponding
scheduler that leads to acceptance with the probability 1. On the
other hand, if the Turing machine does not accept its input then after
some finite number of transitions in the Markov decision process,
either the scheduler ``cheats'' (does not follow the Turing machine
transition function or cell numbering) or we get to a rejecting
state. In both cases, the probability of acceptance is less than 1.
\end{proof}

Table \ref{table:probmc} summarises the known results and the results
from this paper (in bold) on probabilistic systems.  An asterisk
indicates bounds that are not known to be tight.  Note that for the
more complex verification problems, from strategy synthesis for MDPs
onwards, all problems are 2EXP-complete.  Intuitively the complexity
of the model overwhelms the difference in the respective logics.  Similarly,
we see 
that in the stutter-free case the extra succinctness of
$\FOtwoLT$ comes at ``no cost'' over $\TL$--- at least, for the complexity
classes we consider, and where we can establish tight bounds, the respective columns
are identical.

\begin{figure}[h]
\begin{center}
\begin{tabular}{|l|c|c|c|c|c|}
\hline
  & $\TL$ & UTL       & $\FOtwoLT$&$\FOtwo$ & LTL \\
\hline
Markov chain           & \CC \#P & PSPACE &  \CC \#P& \CC PEXP                   &PSPACE\\
HMC & PSPACE${}^*$ & PSPACE &  \CC PSPACE${}^*$  &\CC PEXP &  PSPACE  \\
RMC & \CC PSPACE${}^*$ &  EXPSPACE${}^*$ &\CC PSPACE${}^*$   & \CC EXPSPACE${}^*$ &  EXPSPACE${}^*$\\
MDP $(\forall)$    &  \CC co-NP &\CC EXP & \CC co-NP & \CC co-NEXP &2EXP\\
MDP $(\exists)$   & 2EXP & 2EXP & \CC 2EXP & \CC 2EXP &2EXP\\
MDP (quant)              & 2EXP & 2EXP & \CC 2EXP& \CC 2EXP &2EXP\\
$2\frac{1}{2}$-game $(\forall)$              & 2EXP & 2EXP & \CC 2EXP& \CC 2EXP &2EXP\\
\hline
\end{tabular}
\end{center}
\label{table:probmc}
\end{figure}

\section{Model checking $\FOtwoLTL$}
\label{sec:fo2ltl}
We now turn to combining $\FOtwo$ with automata-based techniques for
LTL, examining verification of the hybrid language $\FOtwoLTL$. As was
done with $\FOtwo$, we first show that we can translate $\FOtwoLTL$
into temporal logic with exponential blow-up in the size of the
formula, giving a simple upper bound. While for $\FOtwo$ the
translation was to unary temporal logic, in this case we have a
translation to $\LTLlet$.

We can look at every $\FOtwoLTL$ formula as being rewritable using
let definitions such that every let definition involves either a
pure $\FOtwo$ formula or a pure $\LTL$ operator.  We get this form
by introducing a let definition for every subformula with one free variable. 
For example, rewriting the formula $\varphi = $ $((\exists y \, (suc(x,y) \wedge P_1(x))) \mathrel{\mathcal{U}} P_0)(x)$ with \emph{let definitions} yields
\begin{table}[h!]

$\begin{array}{rl}

\varphi_{\, \olet} = & \olet ~R_0(x) ~\obe~ P_0(x) ~\oin~ \\ 
 & \olet ~R_1(x) ~\obe~ P_1(x) ~\oin~ \\
 & \olet~ R_2(x) ~\obe~ \exists y \, (suc(x,y) \wedge R_1(x)) ~\oin ~ \\
 &(R_2 \mathrel{\mathcal{U}} R_0)(x)
\end{array}$
\end{table}

Note that although the above uses a combination of $\FOtwo$ and
$\LTL$, each individual definition is either ``pure $\FOtwo$'', or
``pure $\LTL$'', and we can apply the translation of $\FOtwo$ to
$\UTL$ in Lemma \ref{thm:fo2utl} to each $\FOtwo$ definition.  This
gives the following result:

\begin{lem}\label{lem:fo2ltl2ltl}
Given an $\FOtwoLTL$ formula $\varphi$, we can convert it to an equivalent $\LTLlet$ formula $\psi$ such that $|\psi| = \rm{O(2^{|\varphi|^2}})$.
\end{lem}

We could then translate the let definitions away for $\LTL$, to get an
ordinary $\LTL$ formula---thus showing that $\FOtwoLTL$ and $\LTL$
have the same expressiveness.  However, there is no need to perform
this second transformation to get a bound on the complexity of model
checking.  Let definitions do not increase complexity for model
checking $\LTL$, since non-deterministic B\"uchi automata for $\LTL$
and $\LTLlet$ have the same asymptotic size:

\begin{lem}
\label{lemma:ltl_let}
Given an $\LTLlet$ formula $\varphi$, there is an unambiguous B\"uchi automaton $A$ with at most 
$\rm{O(2^{|\varphi|^2})}$ states accepting exactly the language $\{ w \in \Sigma^\omega : w \models \varphi \}$. Moreover this automaton can be constructed in polynomial time in its size. 
\end{lem}

This follows from the fact that the number of 
subformulas of  $\LTLlet$ formulas is linear in the 
formula size
(Lemma~\ref{lem:temp_closure}) and from the following
result of Couvreur et al:

\begin{lem}[\cite{ltltosepaut}]
\label{lemma:ltl_sub}
Given an $\LTL$ formula $\varphi$, there is an unambiguous B\"uchi automaton $A$ with at most 
$\text{O}(|\Sigma||\subf(\varphi)|2^{|\subf(\varphi|)})$ states accepting exactly the language $\{ w: w \in \Sigma^\omega \wedge w \models \varphi \}$. Moreover this automaton can be constructed in polynomial time in its size. 
\end{lem}

As a corollary of Lemmas \ref{lem:fo2ltl2ltl} and \ref{lemma:ltl_let}
we see that we can convert from an $\FOtwoLTL$ formula to an
unambiguous B\"uchi automaton in doubly exponential time, giving a
doubly-exponential bound on the complexity of model-checking.
However, just as in the previous section, we show that we can do
better by direct analysis than via this translation approach.

We begin by looking at the translation given in Lemma
\ref{lem:fo2ltl2ltl} from a different perspective.  Let us extend the
set of atomic propositions $\mathcal{P}$ and alphabet $\Sigma =
2^{\mathcal{P}}$ by adding new atomic propositions $\mathcal{R}$ for
every predicate created in that translation.  Thus we have an extended
alphabet $\Sigma' = 2 ^ {\mathcal{P} \cup \mathcal{R}}$. There is an
obvious restriction mapping taking an infinite word $w'$ over
$\Sigma'$ to a word over $\Sigma$, simply by ignoring all propositions
in $\mathcal{R}$; we denote this by $\text{restrict}(w', \Sigma)$.

\begin{lem}
\label{decomposition}
Given an $\FOtwoLTL$ formula $\varphi$ alphabet $\Sigma$, there is an
$\FOtwo$ formula $\varphi_F$ and an $\LTL$ formula $\varphi_L$ over
$\Sigma'$ having the following two properties for all
$w \in \Sigma^\omega$: (i)~if $w \models \varphi$ then there is a
unique extension $w'$ of $w$ such that
$w' \models \varphi_L \wedge \varphi_F$; (ii)~if
$w \not \models \varphi$ then there is no extension to $w'$ such that
$w' \models \varphi_L \wedge \varphi_F$.  Moreover, $|\varphi_L|,
|\varphi_F| = \rm{O}(|\varphi|^2)$
\end{lem}
\begin{proof}
We use the translation in Lemma \ref{lem:fo2ltl2ltl}, but
consider it simply returning the collection of
let definitions. Corresponding to each
definition is a conjunct stating that
$R_i$ holds iff $\varphi_i$ holds. 
We now examine the form of this conjunct.

Each $\varphi_i$ is either a basic two-variable formula or an $\LTL$
atomic formula.  If $\varphi_i$ is in $\LTL$ then the iff can be
expressed again in $\LTL$: $\BOX (R_i \leftrightarrow \varphi_i)$.  If
$\varphi_i$ is in $\FOtwo$ then the iff above can be expressed as
$\forall x. (R_i(x) \leftrightarrow \varphi_i(x))$.  We can simply let
$\varphi_F$ be th $\FOtwo$ conjuncts and $\varphi_L$ be the $\LTL$
conjuncts to obtain the desired conclusion.

The upper bounds for lengths $|\varphi_L|$ and $|\varphi_F|$ follow
from the fact that $k \le |\varphi|$ and $|\varphi_i| \le |\varphi|$.
\end{proof}

For the formula from the example at the beginning of this section we get following formulas $\varphi_L$ and $\varphi_F$ over $\Sigma'$:
\begin{eqnarray*}
 \varphi_L &=& (R_2 \mathrel{\mathcal{U}} R_0)(x) \wedge \BOX (R_0(x) \leftrightarrow P_0(x)) \wedge \\
 & & \BOX (R_1(x) \leftrightarrow P_1(x)) \\
 \varphi_F &=& \forall x.  (R_2(x) \leftrightarrow \exists y.(\text{suc}(x, y) \wedge R_1(x)))
\end{eqnarray*}

\subsection{Combining automata constructions for $\FOtwo$ and $\LTL$}

Given $\FOtwoLTL$ formula $\varphi$, we can apply Lemma
\ref{decomposition} to obtain an equisatisfiable formula $\varphi_L
\wedge \varphi_F$, where $\varphi_L$ is an $\LTL$ formula and
$\varphi_F$ is an $\FOtwo$ formula over the extended alphabet
$\Sigma'$.  Now we can build a B\"uchi automaton $B_L$ for $\varphi_L$
using the construction from Lemma \ref{lemma:ltl_sub}, as well as a
collection of $2^{2^{\mathit{poly}(|\varphi_F|)}}$ B\"uchi automata
$B_{F_i}$ for $\varphi_F$, using Theorem \ref{thm:FO2_aut}.

For each $i$ we build a product automaton $A_i = B_L \otimes B_{F_i}$
synchronising on the truth values of the newly introduced atomic
propositions $R_i$.  We claim that each product automaton $A_i$ is
unambiguous, the languages they accept are disjoint, and their union
is exactly $\{w\in \Sigma^\omega : w \models \varphi \}$.  This
follows from the fact that each word over $\Sigma$ has only one
extension to a word over $\Sigma'$ for which $B_L$ accepts, along with
the fact that the languages accepted by the $B_{F_i}$ are disjoint.

After producing the synchronised cross product, we can restrict the
input alphabet back to $\Sigma$, because the values of all newly
introduced atomic propositions $p_i \in \Sigma' \setminus \Sigma $ are
fully determined by the truth values of atomic predicates $P_i$ and
the relations defined by $\varphi$.

Therefore we get the following theorem:

\begin{thm}
\label{thm:fo2_ltl:prob}
 $\FOtwoLTL$ formula $\varphi$, there is a collection of doubly
exponentially many (in $|\varphi|$) generalized B\"uchi automata
$A_i$, each of exponential size in $|\varphi|$, such that the
languages they accept are disjoint and the union of these languages is
exactly $\{w\in \Sigma^\omega : w \models \varphi \}$. Moreover, each
automaton $A_i$ is unambiguous and can be constructed by a
non-deterministic Turing machine in polynomial time in its size.
\end{thm}

This translation will now allow us to read off bounds for many
$\FOtwoLTL$ verification problems.

\subsection{Model Checking $\FOtwoLTL$}

Comparing Theorem \ref{thm:FO2_aut} with Theorem \ref{thm:fo2_ltl:prob}, we can easily see 
that automata for $\FOtwo$ in isolation and $\FOtwoLTL$ have the same asymptotic size. 
We can therefore use all automata-based bounds on verification results for $\FOtwo$, provided
that they rely only on unambiguity of the resulting automata.
This allows us to replace $\FOtwo$ with $\FOtwoLTL$ in the results of the
previous sections, giving the following:

\begin{prop}
Model checking $\FOtwoLTL$ properties on Kripke structures,
hierarchical and recursive state machines is in the complexity class NEXP.
\end{prop}

\begin{prop} \label{thm:mcfo2mcfo2ltl}
The threshold problem for model checking $\FOtwoLTL$ 
on both Markov chains and hierarchical Markov chains is in PEXP.
\end{prop}

\begin{prop} \label{thm:rmcfo2ltl} The probability of an  $\FOtwoLTL$ formula holding on a recursive Markov chain 
can be computed in EXPSPACE.
\end{prop}

Now let us consider model checking Markov decision processes. 
Recall that in the proof of the corresponding bound
for $\FO2$,  Theorem \ref{thm:qual_mdp_utl_fo2}, we relied on the fact that
the automata are deterministic in the limit. 
Thus our translation for $\FOtwoLTL$ does not  give us the same bounds
as for $\FO2$. And indeed, the corresponding bound for checking whether
all schedulers achieve probability 1 is worse for $\LTL$ in this case, namely doubly-exponential.
We will show that we can achieve the same bound as for $\LTL$.

\begin{prop}
\label{thm:qual_mdp_fo2ltl}
Determining whether for all schedulers an $\FOtwoLTL$-formula $\varphi$ holds on
a Markov decision process with probability one is in the complexity class 2EXPTIME.
\end{prop}
\begin{proof}
We will decide the corresponding complement problem which asks whether there exists a
scheduler $\sigma$ such that the probability satisfying $\neg \varphi$ is greater than 0.
By applying the translation from Theorem \ref{thm:fo2_ltl:prob}, we get a collection of doubly-exponentially many automata, each of 
exponential size. We can go through all these automata and check if the probability is greater than $0$ for one of them. 
For each automaton, we make a call to the exponential time algorithm for qualitative model checking B\"uchi automata on MDPs from Courcoubetis and Yannakakis \cite{CY95}.
\end{proof}

The following table summarises the results for $\FOtwoLTL$ from this
paper (in bold) concerning both non-deterministic and probabilistic
systems in the context of results for $\FOtwo$ and LTL alone.  An
asterisk indicates bounds that are not known to be tight.  The table
shows that for the models considered in this paper the complexity of
verifying $\FOtwoLTL$ is the maximum of the respective complexities of
$\FOtwo$ and LTL.

\begin{figure}[h!]
\begin{center}
\begin{tabular}{|l|c|c|c|}
\hline
       & $\FOtwoLTL$ &  $\FOtwo$ & $\LTL$ \\
\hline
Kripke structure          & \CC NEXP     & NEXP          &PSPACE\\
HSM   & \CC NEXP &     NEXP & PSPACE\\
RSM          & \CC  NEXP  &   NEXP & EXPTIME \\
\hline
Markov chain   &  \CC PEXP &   PEXP                   &PSPACE\\
HMC & \CC PEXP  &   PEXP &  PSPACE  \\
RMC &  \CC  EXPSPACE${}^*$ & EXPSPACE${}^*$ &  EXPSPACE${}^*$\\
MDP $(\forall)$    &\CC 2EXP &  co-NEXP &2EXP\\
\hline
\end{tabular}
\end{center}
\end{figure}

\section{The impact of Let definitions on model checking}
\label{sec:let}
In the process of examining two-variable logics and their extensions,
we have utilized results on logics extended with Let definitions.  We
now return to considering the impact of Let for several temporal
logics.  First, we note that model checking $\TLlet$, $\UTLlet$ and
$\LTLlet$ properties on both non-deterministic (Kripke structures,
HSMs, RSMs) and probabilistic systems (Markov chains, HMCs, RMCs, MDPs
($\forall$)) has similar computational complexity as for the
corresponding logics without let definitions. We get these results by
simply substituting let definitions to obtain formulas in the base
logic, and then analyze the complexity of model-checking the resulting
formulas.

In the case of $\LTLlet$, we have already noted that the size of the
automaton for $\LTL$ is exponential only in the number of subformulas
(see, e.g. Couvreur et. al. \cite{ltltosepaut})---this leads to Lemma
\ref{lemma:ltl_let}.  Similarly, for $\TLlet$ and $\UTLlet$, we get
the corresponding automata of the same asymptotic size as for $\TL$
and $\UTL$ respectively, because their size depends on the number of
subformulas and the operator depth and not directly on the size of the
formula (see translation in Subsection \ref{subsec:utl2ba}).

In the case of $\FOtwoLET$, we can use
Lemma \ref{thm:fo2utl} to translate the formula to $\UTLlet$ and
then use the result above that the sizes of the automata for $\UTL$
and $\UTLlet$ formulas of the same length are asymptotically equal.
Moreover, since $\LTLlet$ and $\FOtwoLET$ have unambiguous B\"uchi
automata of equal asymptotic size as for $\LTL$ and $\FOtwo$
respectively, we can combine them in the same way as in the proof of
Theorem \ref{thm:fo2_ltl:prob} to get the same complexity upper bounds
for model checking $\FOtwoLTLlet$ as for $\FOtwoLTL$.  Thus we have:

\begin{prop}
For $\LTL$, $\FOtwo$, $\UTL$, $\TL$, and $\FOtwoLTL$, all the upper bounds previously
shown hold also in the presence of Let definitions.
\end{prop}

Finally, we will show that, in contrast to the cases above, the
complexity of model checking $\FOtwoLTlet$ is exponentially worse than
that of $\FOtwoLT$ on both non-deterministic and probabilistic
systems. Thus this is the only logic we have considered where the
introduction of let definitions makes a difference in the
computational complexity of model checking. The following two theorems
show the lower bounds on the complexity of model checking
$\FOtwoLTlet$, which match exactly the upper bounds for $\FOtwoLET$
(compare with Proposition \ref{npbound}).

\begin{prop}
The satisfiability of a $\FOtwoLTlet$ formula under the unary alphabet
restriction is NEXP-hard.
\label{FOtwoLTlet_satisfiability_hardness}
\end{prop} 

\begin{proof}
The proof is by reduction from the halting problem of a
non-deterministic EXPTIME Turing machine $T$ on a given input $I$.
Let $\Gamma$ and $Q$ be respectively the tape alphabet and set of
control states of $T$.  We consider infinite strings over alphabet
\[ \Sigma := (\{P_0, P_1, \ldots P_{2n-1}\} \times 
\{ \Gamma \cup (\Gamma \times Q)\}) \cup \{ \# \} \, .\]
An infinite word $u \in \Sigma^\omega$ encodes a computation of $T$ as
follows.  Each configuration is encoded in a block of contiguous
letters in $u$, with successive configurations arranged in successive
blocks.  Each such block comprises $2^n$ symbols denoting the contents
of each tape cell in the configuration.  A symbol encoding a tape cell
consists of: a letter from $\Gamma \cup (\Gamma \times Q)$ to denote
the contents of the tape cell and whether the read head of the Turing
Machine is currently on the cell (and if so, the current control state
of $T$), and a predicate $P_i$ denoting the address of the tape cell
and the configuration number. Here we use the power of Let definitions
to transform the sequence of $2n$ predicates to values of $2n$-bit
counter (see the proof of Lemma \ref{lem:FO2long}), which represent
the address of configuration and tape cell.  Having thus encoded a
computation of $T$ in a finite prefix of $u$ we require that the
remaining infinite tail of $u$ be the string $\#^\omega$.

We can use short $\FOtwoLTlet$ formulas to identify the position in
the string representing the previous or next position of the tape cell
in the same configuration.  We can also use such formulas to identify
the same position of the tape cell in the previous or next
configuration. Thus we can easily check if the tape symbols are
consistent with the transition function of~$T$. Finally, we ensure $T$
is in the accepting state in the last configuration.
\end{proof}

\begin{prop}
\label{fo2_let_mc_hard}
The decision problem of whether a Markov chain $\M$ satisfies an
 $\FOtwoLTlet$-formula $\varphi$ with probability greater than $1/2$ is
 PEXP-hard.
\end{prop} 

\begin{proof}
The proof is by reduction from the problem of whether a strict
majority of computation paths of a given non-deterministic EXPTIME
Turing machine $T$ on a given input $I$ are accepting.  Without loss
of generality we can assume that any non-halting configuration of $T$
has exactly two successors and that all computations of $T$ on input
$I$ make exactly $2^n$ steps, where $n$ is the length of $I$.

The basic idea, following the proof of NEXPTIME-hardness of
satisfiability for $\FOtwoLTlet$, is to encode computations of $T$ as
strings.  We can define an $\FOtwoLTlet$ formula that is satisfied by
a word $u \in \Sigma^\omega$ precisely when $u$ encodes a legitimate
computation of $T$ on input $I$ according to the encoding scheme used
in Proposition ~\ref{FOtwoLTlet_satisfiability_hardness} Indeed, the
definition is just as described in the proof of NEXPTIME-hardness for
$\FOtwoLTlet$ satisfiability in
Proposition~\ref{FOtwoLTlet_satisfiability_hardness}.

The Markov Chain $\M$ in our reduction is constructed from two copies
of a component $\M'$. The definition of $\M'$ is very simple; it
consists of a directed clique augmented with a single sink state.  In
detail, there is a state $s_\sigma$ for each letter $\sigma \in
\Sigma$; $s_{\#}$ is a sink that makes a transition to itself with
probability $1$; the next-state distribution from $s_\sigma$, $\sigma
\neq \#$, is given by a uniform distribution over all states; finally,
the label of state $s_\sigma$ is $\sigma$.
 
 The Markov chain $\M$ consists of two disjoint copies
 $\M_{\mathit{left}}$ and $\M_{\mathit{right}}$ of $\M'$ that are
 identical except that their states are distinguished by propositions
 $P_{\mathit{left}}$ and $P_{\mathit{right}}$.  The initial state of
 $\M$ is a uniform distribution over all states.

We can partition $\Sigma^\omega$ into three sets $N$, $A$ and $R$,
respectively comprising those strings that don't encode computations
of $T$ on input $I$, those strings that encode accepting computations,
and those strings that encode rejecting computations.  Moreover each
of these sets is definable in $\FOtwoLTlet$ by formulas $\varphi_N$,
$\varphi_A$ and $\varphi_R$ respectively.  

We define the formula $\varphi$ by 

\[ \varphi := ((\forall x\, P_{\mathit{left}}(x)) \wedge 
                      (\varphi_N \vee \varphi_A)) \vee
              ((\forall x\, P_{\mathit{right}}(x)) \wedge \varphi_A) \, .\]
              
To complete the reduction, we claim that $P_\M(L(\varphi)) > 1/2$ if
and only if a strict majority of the computations of Turing Machine
$T$ on input $I$ are accepting.  To see this, observe that if $\M$
produces a trajectory from $N \subseteq \Sigma^\omega$ then that
trajectory is equally likely to have come from $\M_{\mathit{left}}$ or
$\M_{\mathit{right}}$. Using this we can
see that $P_\M(L(\varphi))$ is  $(P_\M(A)+P_\M(N))/2 + P_\M(A)/2$.
Thus $P_\M(L(\varphi)) > 1/2$ iff $2P_\M(A)>1-P_\M(N)$. From this we see that
$P_\M(L(\varphi)) > 1/2$ if and only if $|A| >
|R|$, as required.
\end{proof}

The table below summarises the results for the selected logics. 
An asterisk indicates bounds that are not known to be tight. 
\begin{figure}[h!]
\begin{center}
\begin{tabular}{|l|c|c|c|c|}
\hline
  &  $\TLlet$       & $\FOtwoLTlet$ & $\FOtwoLET$ & $\FOtwoLTLlet$\\
\hline
Kripke structure & NP& NEXP & NEXP & NEXP\\
HSM & NP& NEXP& NEXP & NEXP\\
RSM & NP& NEXP& NEXP & NEXP\\
\hline
Markov chain  &   \#P &   PEXP                   &PEXP & PEXP\\
HMC &  PSPACE${}^*$  &   PEXP &  PEXP  & PEXP\\
RMC &   PSPACE${}^*$ & EXPSPACE${}^*$ &  EXPSPACE${}^*$ & EXPSPACE${}^*$\\
MDP $(\forall)$    & co-NP &  co-NEXP &co-NEXP & 2EXP\\
\hline
\end{tabular}
\end{center}
\end{figure}
\section{Conclusions and ongoing work}
In this paper we have compared the complexity of verifying properties
in the two best-known elementary fragments of monadic first-order
logic on words: LTL and $\FOtwo$.  
 We provided several different
logic-to-automaton constructions that are useful for verification
of $\FOtwo$. One translations allows us to understand
the complexity of verifying full $\FOtwo$ via analysis of
unary temporal logic; a second is useful for the sublanguage
of $\FOtwo$ with only the linear-ordering; the third is
useful for getting deterministic automata, which is needed
for obtaining bounds for certain game-related problems.
We have shown that these translations put
together allow us to understand
the  complexity
 of verification and synthesis problems for both
non-deterministic and probabilistic models transition systems,
including those arising from hierarchical and recursive state
machines.

While LTL is more expressive than $\FOtwo$, $\FOtwo$ can be
exponentially more succinct.  We have shown that the effect of these
opposing factors on the complexity of model checking depends on the
model, e.g., $\FOtwo$ has higher complexity on Markov chains while
$\LTL$ has higher complexity on MDPs.  By contrast, in the
stutter-free case the extra succinctness of $\FOtwoLT$ comes for
free---all verification problems have the same complexity as for
$\TL$.  For the most structured models e.g., two-player games and
quantitative verification of MDPs, the complexity of the model
dominates any difference in the logics.

We are currently examining the succinctness of Let definitions when
added to each of our logics. A number of succinctness
results can be found in
this work, but 
we have left open the succinctness of Let in certain
situations, e.g., for the logic $\FOtwoLTL$.
Finally, we are investigating the extension of the
techniques introduced here from
words to  trees.

{\bf Acknowledgments:} M.~Benedikt is supported in part by EPSRC
grants EP/G004021/1 and EP/H017690/1.  Worrell is supported in part by
EPSRC grant EP/G069727/1.

\bigskip
\bibliographystyle{alpha}
\bibliography{litb}
\end{document}